\documentclass[pdflatex,sn-apa]{sn-jnl}

\usepackage{amsthm, amsfonts, amsmath,amssymb,stmaryrd}
\usepackage{enumitem}
\usepackage{graphicx,caption,tikz,comment,subcaption}
\usepackage{multirow}%
\usepackage{mathrsfs}%
\usepackage[title]{appendix}%
\usepackage{xcolor}%
\usepackage{textcomp}%
\usepackage{manyfoot}%
\usepackage{booktabs}%
\usepackage{algorithm}%
\usepackage{algorithmicx}%
\usepackage{algpseudocode}%
\usepackage{listings}%
\usepackage{hyperref}

\usetikzlibrary{arrows.meta, positioning,calc, fit}

\theoremstyle{thmstyleone}%
\newtheorem{theorem}{Theorem}
\newtheorem{lemma}[theorem]{Lemma}%
\newtheorem*{claim*}{Claim}
\newtheorem*{uconj}{Conjecture}

\theoremstyle{thmstyletwo}%
\newtheorem{example}{Example}%
\newtheorem{remark}{Remark}%

\theoremstyle{thmstylethree}%
\newtheorem{definition}{Definition}%

\newenvironment{claimproof}{\begin{proof}[Proof of claim]}{\phantom\qedhere\hfill$\blacksquare$\end{proof}}

\newcommand{\CaseAState}[4]{%
  \node[world] (#1) at #2 {};
  \node[tlabel] at ($(#1)+(0,0.24)$) {$#3$};
  \node[vlabel] at ($(#1)+(0,0.66)$) {$#4$};
}

\newcommand{\CaseARoot}[4]{%
  \node[world] (#1) at #2 {};
  \node[below=1pt of #1] {$w$};
  \node[right=2pt of #1] {$#3$};
  \node[left=4pt of #1] {$#4$};
}

\title{Classification of $\sigma$-validity in iterated announcements}
\date{\today}
\author*[1]{\fnm{Eiji} \sur{Yamada}}\email{yamada.e.112d@m.isct.ac.jp}
\affil*[1]{\orgdiv{Department of Mathematical and Computing Science}, \orgname{Institute of Science Tokyo}, \state{Tokyo}, \country{Japan}}

\abstract{In their 2018 paper, \AA gotnes, van Ditmarsch, and Wang extended the notions of success and self-refutation in public announcements to true lies, impossible lies, and $\sigma$-validity in general. Here, $\sigma$ is a finite or infinite sequence of $0$s and $1$s. For example, successful formulas and self-refuting formulas are $11$-valid and $10$-valid, respectively. They then posed a conjecture on the classification of such sequences in terms of $\sigma$-validity. In this paper, we disprove the conjecture and give corrected classifications for multi-agent K45, single-agent KD45, multi-agent KD45 with more than one agent, and multi-agent S5 after reformulating the statement more explicitly. The results indicate that there is an asymmetry between truthful announcements and false announcements: the former are stable while the latter are fragile in general. In particular, all successful formulas remain true forever while some impossible lies can be true at some point when repeatedly announced. Also, although some self-refuting formulas can become true again after following the truth pattern $10$, all $100$-valid formulas are destructive in the sense that they remain false forever once they become false. On the other hand, some true lies are fragile in the sense that truths created by lying can become false again.}
\begin{document}

\maketitle

\section{Introduction}
Dynamic epistemic logic is the study of information flow, knowledge acquisition, and belief revision. Public announcements are one of the main topics in the field. A public announcement of a proposition $\varphi$ means to announce that $\varphi$ holds to a group of agents. \emph{Public announcement logic} (\emph{PAL}) was proposed \citep{plaza2007} to express public announcements, which deletes all the states where $\varphi$ is false after a public announcement of $\varphi$. 

However, in public announcement logic, we can no longer consider the truth of formulas at states that have already been eliminated. That is, when we consider the truth of a formula $\varphi$ at the pointed model $M,w$, we can no longer consider the truth at $w$ in the updated model $M_\varphi$ whenever $M,w\models\lnot\varphi$, simply because $w\notin M_\varphi$.

On the other hand, \cite{gerbrandyGroeneveld1997} proposed the logic of \emph{conscious updates} also known as \emph{believed public announcement logic} (\emph{BPAL}), which deletes all arrows that go to states where $\varphi$ is false. That is, in believed public announcement logic, we define the accessibility relation $R|\varphi$ of the updated model $M|\varphi$ as $w(R|\varphi)v$ iff $wRv$ and $M,v\models\varphi$ where $M=(W,R,V)$ is a (initial) model. This allows us to consider not only truthful announcements but also false announcements. In particular, we can keep track of the truth of a formula $\varphi$ at $M,w$ that becomes false at some stage when $\varphi$ is announced repeatedly (\emph{iterated announcements} of $\varphi$).

One of the major topics in public announcements is the discussion of successful formulas. We say that a formula $\varphi$ is \emph{successful} iff whenever $\varphi$ is true, it remains true after a public announcement of $\varphi$. Formally, $\varphi$ is successful iff the formula $\varphi\to[\uparrow\varphi]\varphi$ is valid, where $[\uparrow\varphi]\psi$ means in BPAL that after a public announcement of $\varphi$, $\psi$ holds. This notion matters since ``successful'' guarantees that true information is shared with others without changing its truth, which is oftentimes the purpose of announcements.

In contrast, there are formulas that change their truth after an announcement. We say that $\varphi$ is \emph{self-refuting} iff whenever $\varphi$ is true, it becomes false after an announcement~\citep{hollidayIcard2010}, with the \emph{Moore sentence} $\varphi:=p\land\lnot\Box p$ being the most typical example~\citep{Moore1952-MOOART-2,hintikka1962}. \cite{agotnesVanDitmarschWang2018} proposed \emph{true lies} and \emph{impossible lies} as the remaining cases, where $\varphi$ is a true lie iff whenever $\varphi$ is false, it becomes true after an announcement, and $\varphi$ is an impossible lie iff whenever $\varphi$ is false, it remains false after an announcement (see Table~\ref{tab:four_notions}).

\begin{table}[htbp]
  \centering
  \begin{tabular}{|c|c|}
    \hline
    Successful & $\varphi\to [\uparrow\varphi]\varphi$\\\hline
    Self-refuting & $\varphi\to [\uparrow\varphi]\lnot\varphi$\\\hline
    True lie & $\lnot\varphi\to [\uparrow\varphi]\varphi$\\\hline
    Impossible lie & $\lnot\varphi\to [\uparrow\varphi]\lnot\varphi$\\\hline 
  \end{tabular}
  \caption{Four notions}
   \label{tab:four_notions}
\end{table}
In their analysis of true lies, \cite{agotnesVanDitmarschWang2018} also discussed iterated announcements in BPAL, and defined the notion of $\sigma$-\emph{validity} and $\sigma$-\emph{satisfiability} for a finite or infinite sequence $\sigma$ of $0$s and $1$s with length$\,\geq 2$. A formula $\varphi$ is $\sigma$-valid iff for all pointed models $M,w$, the truth values of $\varphi$ at $w$ under iterated announcements of $\varphi$ exactly follow $\sigma$ whenever the truth value of $\varphi$ at $M,w$ is the first digit of $\sigma$. For example, $11$-valid, $10$-valid, $01$-valid, and $00$-valid are exactly the same as successful, self-refuting, true lies, and impossible lies, respectively. So, the notion of $\sigma$-validity is a generalization of these four notions. We say that $\varphi$ is $\sigma$-satisfiable iff there is an initial pointed model $M,w$ that indeed realizes the truth pattern of $\sigma$ under iterated announcements of $\varphi$. This notion is used to exclude trivial cases. For example, the formula $\top$ is $01$-valid since $\lnot\top\to[\uparrow\top]\top$ is vacuously true, but not $01$-satisfiable since no pointed model $M,w$ realizes the truth pattern $01$. Furthermore, we say that $\varphi$ is \emph{non-trivially} $\sigma$-\emph{valid} iff $\varphi$ is both $\sigma$-valid and $\sigma$-satisfiable.

In \cite{agotnesVanDitmarschWang2018}, \AA gotnes, van Ditmarsch, and Wang posed the following conjecture: 
\begin{uconj}[{\citealp[Conjecture~13]{agotnesVanDitmarschWang2018}}]\label{Conjecture13}
    The following completion equivalence classes of $\sigma$-valid formulas are all non-empty and all different, and there are no other classes.
    \begin{equation*}
01^k \,(k>0) \qquad 10^k \,(k>0) \qquad 01^k0 \,(k\geq 0) \qquad 10^k1\, (k\geq 0)
    \end{equation*}

\end{uconj}
Here, the definition of \emph{completion equivalence} was given in the paper as follows:
\begin{quote}
  Given $\sigma\in \{0,1\}^\omega$ and prefix $\tau$ of $\sigma$, then $\sigma$ and $\tau$ (and $\tau$ and $\sigma$) are \emph{completion equivalent} iff for all intermediate strings $\tau'$ (for all prefixes $\tau'$ of $\sigma$ such that $\tau$ is a prefix of $\tau'$), a formula is $\sigma$-valid iff it is $\tau'$-valid. This relation is clearly an equivalence
relation. The shortest sequence of a completion equivalence class is the \emph{representative} and the longest (possibly infinite) sequence of this class is the \emph{completion}.
\end{quote}
However, this definition needs a slight modification since the relation is defined only between a fixed sequence and its prefixes, and the underlying set is unclear. Also, the meanings of ``non-empty'', ``different'', and ``there are no other classes'' in the conjecture statement are ambiguous. They further mentioned, ``we have not investigated systematically which of the conjectured $\sigma$-valid types are non-empty and for which classes of models.''\footnote{At the beginning of the section which contains Conjecture 13, they stated, ``In this section we present results and conjectures for $\sigma$-satisfiable and $\sigma$-valid formulas,
with respect to the classes K and KD45.'' However, the intended class for Conjecture 13 was not explicit.} 

In this paper, we show that the conjecture is false, and give a corrected version for multi-agent K45, single-agent KD45, multi-agent KD45 with more than one agent, and multi-agent S5 (Theorem~\ref{thm:classification theorem}) after making the relevant definitions and the conjecture statement more explicit (Definitions~\ref{def:sigma_satisfiable_and_valid} and \ref{def: equivalence_relation_on_sigma_validity}).The choice of the classes of frames K45, KD45, and S5 is to see the roles of belief consistency and factivity.

The structure of this paper is as follows.
Section~\ref{sec:Basic_definitions} gives definitions of formulas, models, and other basic concepts. Section~\ref{sec:Classification_of_sigma-validity} defines the notions of $\sigma$-satisfiability and $\sigma$-validity proposed in \cite{agotnesVanDitmarschWang2018}, and gives the corrected conjecture. Section~\ref{sec:Lemmas_for_the_classification_theorem} lists multiple lemmas for our theorem, categorizing them into collapse lemmas, nonexistence lemmas, and existence lemmas. Section~\ref{sec:Discussions} discusses the interpretation of the classification, as well as a comparison between PAL and BPAL, and possible future directions.  

A Lean formalization is available at \url{https://github.com/eiyamada/lean-paper-formalizations}.

\section{Basic definitions}\label{sec:Basic_definitions}
We first list several basic definitions. Let $\mathbf{Prop}$ be a countably infinite set of proposition letters.
\begin{definition}
  Define \textit{formulas} in multi-agent epistemic logic $\mathcal{L}$ by
  \begin{equation*}
    \varphi::=p\mid\lnot\varphi\mid\varphi\land\psi\mid\Box_i\varphi\quad(p\in \mathbf{Prop},\quad i\in G)
  \end{equation*}
\end{definition}
\begin{definition}[{\citealp[Definition~4.4]{vanDitmarschHoekKooi2007}}]
  Define \textit{formulas} in public announcement logic $\mathcal{L}_{PAL}$ by
  \begin{equation*}
    \varphi::=p\mid\lnot\varphi\mid \varphi\land\psi\mid \Box_i\varphi\mid [!\varphi]\psi\quad (p\in\mathbf{Prop},\quad i\in G)
  \end{equation*}
\end{definition}

\begin{definition}[{\citealp[Definition~1]{agotnesVanDitmarschWang2018}}]
  Define \textit{formulas} in believed public announcement logic $\mathcal{L}_{\text{BPAL}}$ by
  \begin{equation*}
    \varphi::=p\mid\lnot\varphi\mid \varphi\land\psi\mid \Box_i\varphi\mid [\uparrow\varphi]\psi\quad (p\in\mathbf{Prop},\quad i\in G)
  \end{equation*}
\end{definition}
It is known that the expressive powers of the three languages above are the same (for $\mathcal{L}$ and $\mathcal{L}_{PAL}$, see Chapter 8 of \cite{vanDitmarschHoekKooi2007}, and for $\mathcal{L}$ and $\mathcal{L}_{BPAL}$, see section 2 of \cite{agotnesVanDitmarschWang2018}). Hereafter when we simply say ``a formula'', we usually refer to a formula in $\mathcal{L}$.
\begin{definition}
  A \textit{model} is a tuple $M=(W,\{R_i\}_{i\in G},V)$ where $W$ is a non-empty set of \textit{states}, for each $i\in G$, $R_i$ is a binary relation on $W$ (\textit{accessibility relation}), and $V\colon \mathbf{Prop}\to 2^W$ is a \textit{valuation function}.
\end{definition}

Let $R^*$ be the reflexive transitive closure of a binary relation $R$.
\begin{definition}
  A model $M'=(W',\{R_i'\}_{i\in G},V')$ is a \textit{submodel} of $M=(W,\{R_i\}_{i\in G},V)$ (written $M'\subseteq M$) iff $W'\subseteq W$, $R_i'=R_i\cap (W'\times W')$ for all $i\in G$, and $V'(p)=V(p)\cap W'$ for all $p\in\mathbf{Prop}$. The \textit{generated submodel} of $M$  at $w$ is the submodel $M_w=(W_w,\{R_{iw}\}_{i\in G},V_w)$ with $W_w=\{v\in W\colon w\left(\bigcup_{i\in G}R_i\right)^*v\}$.
\end{definition}

For a formula $\varphi$ and a pointed model $M,w$, let $\llbracket\varphi\rrbracket_M:=\{w\in W\colon M,w\models\varphi\}$ be the set of states in which $\varphi$ holds.
\begin{definition}\label{def:relativized_model}
  Let $M=(W,\{R_i\}_{i\in G},V)$ be a model and $\varphi$ be a formula.
  \begin{itemize}
    \item The \emph{relativization} of $M$ to $\varphi$ under public announcement is the model $M_{|\varphi}=(W_{|\varphi}, \{R_{i|\varphi}\}_{i\in G}, V_{|\varphi})$ where $W_{|\varphi}=\llbracket\varphi\rrbracket_M$, $R_{i|\varphi}=R_i\cap (W_{|\varphi}\times W_{|\varphi})$, and $V_{|\varphi}(p)=V(p)\cap W_{|\varphi}$.
    \item The \emph{relativization} of $M$ to $\varphi$ under believed public announcement is the model $M|\varphi=(W, \{R_i|\varphi\}_{i\in G}, V)$ with $R_i|\varphi=R_i\cap (W\times \llbracket \varphi\rrbracket_M)$.
  \end{itemize}
\end{definition}
\begin{definition}
  Let $M=(W,\{R_i\}_{i\in G},V)$ be a model. Define truth as follows.
  \begin{enumerate}
    \item $M,w\models p\iff w\in V(p)$.
    \item $M,w\models\lnot\varphi\iff M,w\not\models\varphi$.
    \item $M,w\models\varphi\land\psi\iff M,w\models\varphi\text{ and }M,w\models\psi$.
    \item $M,w\models\Box_i\varphi\iff M,v\models\varphi$ for all $v$ with $wR_iv$.
    \item $M,w\models[!\varphi]\psi\iff M,w\models\varphi$ implies $M_{|\varphi}, w\models\psi$.
    \item $M,w\models[\uparrow\varphi]\psi\iff M|\varphi, w\models\psi$.
  \end{enumerate}
\end{definition}
We now define four notions in public announcements~\citep{agotnesVanDitmarschWang2018}.
\begin{definition}Let $\varphi$ be a formula.
  \begin{itemize}
    \item $\varphi$ is \textit{successful} iff $\varphi\to[\uparrow\varphi]\varphi$ is valid.
    \item $\varphi$ is \textit{self-refuting} iff $\varphi\to[\uparrow\varphi]\lnot\varphi$ is valid.
    \item $\varphi$ is a \textit{true lie} iff $\lnot\varphi\to[\uparrow\varphi]\varphi$ is valid.
    \item $\varphi$ is an \textit{impossible lie} iff $\lnot\varphi\to[\uparrow\varphi]\lnot\varphi$ is valid.
  \end{itemize}
\end{definition}

The following lemma states that in K45, the truth of ``modal atoms'' does not change before and after moving between two states.
\begin{lemma}\label{lem:modal agreement lemma}
  Let $M=(W,\{R_i\}_{i\in G},V)$ be a K45 model. If $wR_iv$, then for any formula $\varphi$ of the form $\Box_i\psi$ or $\Diamond_i\psi$, we have
  \begin{equation*}
    M,w\models\varphi\Longleftrightarrow M,v\models\varphi
  \end{equation*}
\end{lemma}
\begin{proof}
  For $\varphi=\Box_i\psi$, $(\Rightarrow)$ follows from transitivity and $(\Leftarrow)$ follows from Euclideanness. For $\varphi=\Diamond_i\psi$, the order of the properties is reversed.
\end{proof}

\section{Classification of $\sigma$-validity}\label{sec:Classification_of_sigma-validity}
For a finite set $\Sigma$, let $\Sigma^*=\{a_1\cdots a_n\colon n\geq 0\text{ and } a_i\in \Sigma\text{ for all }1\leq i\leq n\}$ be the set of sequences of finite length from $\Sigma$, and $\Sigma^\omega=\{a_1a_2\cdots\colon a_i\in\Sigma\text{ for all }i\geq 1\}$ be the set of sequences of length $\omega$ from $\Sigma$. Furthermore, for each $n\geq 0$, let $\Sigma_{\geq n}=\{\sigma\in\Sigma^*\colon |\sigma|\geq n\}\cup \Sigma^\omega$ be the set of finite or infinite sequences of length $\geq n$ from $\Sigma$.

Given $\sigma\in\{0,1\}^*$ and $1\leq k\leq |\sigma|$, $\sigma_k$ denotes the $k$th digit of $\sigma$ and $\sigma|k$ the prefix consisting of the first $k$ elements of $\sigma$. We abuse notation and also view $\sigma_k$ as a function $\sigma_k\colon \mathcal{L}\to\mathcal{L}$ on formulas such that
\begin{equation*}
  \sigma_k(\varphi)=
  \begin{cases}
    \varphi      & \text{if $\sigma_k=1$} \\
    \lnot\varphi & \text{if $\sigma_k=0$}
  \end{cases}
\end{equation*}
We state Definition~4 of \cite{agotnesVanDitmarschWang2018} (in a slightly simpler but equivalent way).
\begin{definition}\label{def:sigma_satisfiable_and_valid}
  Let $\sigma\in\{0,1\}^*$ with $n=|\sigma|\geq 2$ and $\varphi$ be a formula. $\varphi$ is $\sigma$-\emph{satisfiable} iff the formula
  \begin{equation*}
    \sigma_1(\varphi)\land\bigwedge_{k=2}^n[\uparrow\varphi]^{k-1}\sigma_k(\varphi)
  \end{equation*}
  is satisfiable where $[\uparrow\varphi]^n$ is $[\uparrow\varphi]$ repeated $n$ times. $\varphi$ is $\sigma$-\emph{valid} iff
  \begin{equation*}
    \sigma_1(\varphi)\to\bigwedge_{k=2}^n[\uparrow\varphi]^{k-1}\sigma_k(\varphi)
    \end{equation*}
  is valid.
  For $\sigma\in\{0,1\}^\omega$, $\varphi$ is $\sigma$-\textit{satisfiable} iff there is a model $M,w$ such that for all $n\geq 0$, $M|^n\varphi,w\models\sigma_{n+1}(\varphi)$.
  $\varphi$ is $\sigma$-valid iff $\varphi$ is $\sigma|k$-valid for all $k\geq 2$. $\varphi$ is \textit{non-trivially} $\sigma$-\textit{valid} iff $\varphi$ is both $\sigma$-valid and $\sigma$-satisfiable.
\end{definition}
\begin{example}\label{ex:sigma-validity}
  If $|\sigma|=2$, $\varphi$ is $\sigma$-valid iff $\sigma_1(\varphi)\to[\uparrow\varphi]\sigma_2(\varphi)$ is valid. In particular, $\varphi$ is 01-valid iff $\lnot\varphi\to[\uparrow\varphi]\varphi$ is valid (i.e., $\varphi$ is a true lie).

  The formula $p\lor\Box p$ is $01^\omega$-valid in K45: Suppose that $M,w\models\lnot(p\lor\Box p)$. Then, we have $M,w\models \lnot p\land\lnot\Box p$. Take any $v\in R|(p\lor\Box p)(w)$. Then, $M,v\models p\lor\Box p$ but by Lemma~\ref{lem:modal agreement lemma}, we also have $M,v\models\lnot\Box p$. Thus, we must have $M,v\models p$, so $M|(p\lor\Box p),w\models\Box p$ and hence $M|^k(p\lor\Box p),w\models(p\lor\Box p)$ for all $k\geq 1$.
Also, $p\lor\Box p$ is $01^\omega$-satisfiable in K45 since for the S5 (hence K45) model $M=(\{w,v\}, W\times W, V)$ with $V(p)=\{v\}$, we have $M,w\models\lnot (p\lor\Box p)$ and $M|^k(p\lor\Box p),w\models p\lor\Box p$ for all $k\geq 1$. Therefore, $p\lor\Box p$ is non-trivially $01^\omega$-valid in K45\footnote{$p\lor\Box p$ (``$p$ or an agent believes $p$'') is considered the most fundamental example of a true lie much the same as the Moore sentence $p\land\lnot\Box p$ (``$p$ but an agent does not believe $p$'') is the most fundamental example of a self-refuting formula.}.

  If $|\sigma|=3$, $\varphi$ is $\sigma$-valid iff $\sigma_1(\varphi)\to([\uparrow\varphi]\sigma_2(\varphi)\land[\uparrow\varphi][\uparrow\varphi]\sigma_3(\varphi))$ is valid. In particular, $\varphi$ is $011$-valid iff $\lnot\varphi\to([\uparrow\varphi]\varphi\land[\uparrow\varphi][\uparrow\varphi]\varphi)$ is valid.
\end{example}
Now, let
\begin{equation*}
  S=\{\sigma\in\{0,1\}_{\geq 2}\colon \text{There is a non-trivially }\sigma\text{-valid formula}\}.
\end{equation*}
We modify the notion of completion equivalence as we explained in Introduction.
\begin{definition}\label{def: equivalence_relation_on_sigma_validity}
  For $\sigma\in S$, let $\text{Val}(\sigma)=\{\varphi\in\mathcal{L}\colon \varphi\text{ is }\sigma\text{-valid}\}$. Define the equivalence relation $\sim$ on $S$ by
  \begin{equation*}
    \sigma\sim\tau\iff \text{Val}(\sigma)=\text{Val}(\tau).
  \end{equation*}
\end{definition}
We now state our classification theorem, which is a corrected version of Conjecture 13 in \cite{agotnesVanDitmarschWang2018}.
\begin{theorem}[Classification theorem]\label{thm:classification theorem}\hfill
  \begin{enumerate}[label=\textup{(\arabic*)}]
    \item\label{itm:classification_k45} In multi-agent K45:
          \begin{equation*}
            S=[0^\omega]\sqcup \bigsqcup_{k\geq 1}[01^k]\sqcup[01^\omega]\sqcup[10^\omega]\sqcup[1^\omega].
          \end{equation*}
    \item\label{itm:classification_single_kd45} In single-agent KD45:
          \begin{equation*}
            S=[0^\omega]\sqcup \bigsqcup_{k\geq 1}[01^k]\sqcup[01^\omega]\sqcup[10]\sqcup[10^\omega]\sqcup[101^\omega ]\sqcup[1^\omega].
          \end{equation*}
          \item\label{itm:classification_multi_kd45} In multi-agent KD45 with $|G|\geq 2$ and multi-agent S5:
\begin{equation*}
            S=[0^\omega]\sqcup\bigsqcup_{k\geq 2}[0^k]\sqcup\bigsqcup_{k\geq 1}[01^k]\sqcup[01^\omega]\sqcup[10]\sqcup[10^\omega]\sqcup[101^\omega]\sqcup[1^\omega].
          \end{equation*}
        \end{enumerate}
\end{theorem}
\begin{remark}
  The single-agent KD45 case differs from the multi-agent K45 case only in that the former has the classes $[10]$ and $[101^\omega]$. The multi-agent KD45 with $|G|\geq 2$ and multi-agent S5 cases differ from the single-agent KD45 case only in that the former have the classes $\bigsqcup_{k\geq 2}[0^k]$. 
  
  We check which of the equivalence classes each sequence of length 2 belongs to. First, $00\in [0^\omega]$ in multi-agent K45 and single-agent KD45 by Lemma~\ref{lem:00-validity_collapse} while $00\in [00]$ in multi-agent KD45 with $|G|\geq 2$ and multi-agent S5. Clearly, $01\in [01]$ in all the classes of frames. Also, $10\in [10^\omega]$ in multi-agent K45 while $10\in [10]$ in multi-agent KD45 and S5 by Lemma~\ref{lem:10k-validity_collapse}. Finally, $11\in [1^\omega]$ by Lemma~\ref{lem:11-validity_collapse}.  
\end{remark}
\begin{proof}
  \begin{enumerate}[label=(\arabic*)]
    \item
          \smallskip
          \noindent\textbf{Well-definedness} We first show that the displayed equivalence classes are well-defined. For $0^\omega$, $\varphi:=\bot$ is $0^\omega$-valid and $0^\omega$-satisfiable, so $0^\omega$ is indeed in $S$. For $1^\omega$, take $\varphi:=\top$. For $10^\omega$, take $\varphi:=p\land\lnot\Box_i p$. This is 10-valid hence $10^\omega$-valid by Lemma~\ref{lem:10k-validity_collapse} and $10^\omega$-satisfiable by the S5 (hence K45) model $M=(W,\{R_i\}_{i\in G},V)$ where $W=\{w,v\}$, $R_i=W\times W$ for all $i\in G$, and $V(p)=\{w\}$. 
          
          For $01^k\,(k\geq 1)$, follows from Lemma~\ref{lem:01k-valid_but_not_01kplus1-valid}. For $01^\omega$, $\varphi:=p\lor\Box_i p$ is non-trivially $01^\omega$-valid by Example~\ref{ex:sigma-validity}.

          \smallskip
          \noindent
          \textbf{Equality} Next, we show the equality. $(\supseteq)$ is clear.
          $(\subseteq)$: Take any $\sigma\in S$. Note that $\sigma$ has length $\geq 2$.

          If $\sigma$ starts with 00, $\sigma\in [0^\omega]$: In fact, let $\varphi$ be any formula.  If $\varphi$ is $\sigma$-valid, then it is $00$-valid hence $0^\omega$-valid by Lemma~\ref{lem:00-validity_collapse}. Conversely, suppose that $\varphi$ is $0^\omega$-valid. Then, $\sigma$ cannot contain a 1, since otherwise contradicts $\sigma$-satisfiability.  Thus, we have $\sigma\in [0^\omega]$.

          If $\sigma$ starts with $11$, $\sigma\in [1^\omega]$ by the same reasoning as above.

          If $\sigma$ starts with 10, we have $\sigma\in [10^\omega]$ by Lemma~\ref{lem:10k-validity_collapse}.

          If $\sigma$ starts with 01, $\sigma$ is either of the form $01^k$ for some $k\geq 1$, or $01^\omega$, or $\sigma$ has a prefix of the form $01^k0$ for some $k\geq 1$. The last case is impossible by Lemma~\ref{lem:nonexistence_of_01k0-validity}. Thus, we have either $\sigma\in [01^k]$ for some $k\geq 1$ or $\sigma\in [01^\omega]$. 

          \smallskip
          \noindent
          \textbf{Disjointness} Finally, we show that the equivalence classes are disjoint by a pair-wise check. For $[0^\omega]$ and $[1^\omega]$, $\Box_i\bot$ is $1^\omega$-valid but not $0^\omega$-valid. For $[0^\omega]$ and $[10^\omega]$, $p$ is $0^\omega$-valid but not $10^\omega$-valid. Disjointness between $[0^\omega]$ and $[01^k]$, $[0^\omega]$ and $[01^\omega]$, and $[1^\omega]$ and $[10^\omega]$ are clear. For $[1^\omega]$ and $[01^k]$, and $[1^\omega]$ and $[01^\omega]$, $p$ is $1^\omega$-valid but not $01^k$-valid nor $01^\omega$-valid. For $[10^\omega]$ and $[01^k]$, and $[10^\omega]$ and $[01^\omega]$, $p\lor\Box_i p$ is $01^k$-valid and $01^\omega$-valid but not $10^\omega$-valid since $p\lor\Box_i p$ is true forever at the K45 model $M=(W,\{R_i\}_{i\in G}, V),w$ where $W=\{w\}$, $R_i=\{(w,w)\}$ for all $i\in G$, and $V(p)=W$. For $[01^k]$ and $[01^\omega]$, follows from Lemma~\ref{lem:01k-valid_but_not_01kplus1-valid}.

    \item
          \smallskip
          \noindent\textbf{Well-definedness} For $[10]$, $\varphi:=p\land\lnot\Box_i p$ is non-trivially $10$-valid. For $[101^\omega]$, follows from Lemmas~\ref{lem:101-validity} and \ref{lem:101-validity_collapse}. The arguments are the same for the other equivalence classes.

          \smallskip
          \noindent
          \textbf{Equality}
          Take any $\sigma\in S$.

          If $\sigma$ starts with 00, 11, or 01, the same argument as \ref{itm:classification_k45} works.

          If $\sigma$ starts with 10, there are three possibilities. If $\sigma$ is 10, $\sigma\in [10]$. If $\sigma$ starts with 100, $\sigma\in[10^\omega]$ by Lemma~\ref{lem:10k-validity_collapse} so $\sigma\in [10^\omega]$. If $\sigma$ starts with 101, then $\sigma\in [101^\omega]$ by Lemma~\ref{lem:101-validity_collapse}.

          \smallskip
          \noindent
          \textbf{Disjointness} We separate the displayed equivalence classes into
          \begin{equation*}
            A=\{[0^\omega], \bigsqcup_{k\geq 1}[01^k], [01^\omega], [1^\omega]\}
          \end{equation*}
          and
          \begin{equation*}
            B=\{[10], [10^\omega], [101^\omega]\}.
          \end{equation*}

          For the disjointness in $A$, we can use the same witness formulas as in \ref{itm:classification_k45}.

          For the disjointness in $B$: For $[10]$ and $[10^\omega]$, there is a 10-valid formula that is not 100-valid by Lemma~\ref{lem:10k-validity_collapse}. For $[10]$, $[101^\omega]$, $p\land\lnot\Box p$ is 10-valid but not $101^\omega$-valid so they are disjoint. For $[10^\omega]$ and $[101^\omega]$, $p\land\lnot\Box p$ is $10^\omega$-valid but not $101^\omega$-valid.

          For the disjointness between $A$ and $B$: For the elements $[10],[10^\omega],[101^\omega]$ in $B$ and the elements $[0^\omega],[1^\omega]$ in $A$, $p$ is both $0^\omega$-valid and $1^\omega$-valid but not $10$-valid. For the elements $[10],[10^\omega],[101^\omega]$ in $B$ and $[01^k],[01^\omega]$ in $A$, $p\lor\Box p$ is both $01^k$-valid and $01^\omega$-valid but not $10$-valid. Therefore, the elements in $A$ and the elements in $B$ are disjoint.

    \item \smallskip\noindent
          \textbf{Well-definedness} For $[0^k]\,(k\geq 2)$, it follows from Lemmas~\ref{lem:0k_but_not_0kplus1_multi_kd45} and \ref{lem:0k-valid_but_not_0_plus1-validity} .  The arguments are the same for the other equivalence classes.

          \smallskip\noindent\textbf{Equality} Take any $\sigma\in S$.

          If $\sigma$ starts with 0, $\sigma$ is in one of the following equivalence classes:
          \begin{equation*}
            [0^k]\;(k\ge 2),\qquad [0^\omega],\qquad [01^k]\;(k\ge 1),\qquad [01^\omega].
          \end{equation*}
          In fact, if $\sigma$ starts with $01$, it is either in $[01^k]$ for some $k\geq 1$ or $[01^\omega]$ since there is no non-trivially $01^k0$-valid formula for each $k\geq 1$ by Lemmas~\ref{lem:nonexistence_of_01k0-validity} and \ref{lem:nonexistence_of_01k0-validity_multi_kd45}. Also, if $\sigma$ starts with $00$, it is either in $[0^k]$ for some $k\geq 2$ or $[0^\omega]$ since there is no non-trivially $0^k1$-valid formula for each $k\geq 2$ by Lemmas~\ref{lem:nonexistence_of_0k1-validity} and \ref{lem:nonexistence_of_0k1-validity_multi_kd45}.

          If $\sigma$ starts with 1, $\sigma$ is in one of the following equivalence classes:
          \begin{equation*}
            [10],\qquad [10^\omega],\qquad [101^\omega], \qquad [1^\omega].
          \end{equation*}
          In fact, if $\sigma$ starts with 11, $\sigma\in [1^\omega]$ by Lemma~\ref{lem:11-validity_collapse}. If $\sigma$ is $10$, $\sigma\in [10]$. If $\sigma$ starts with 101, $\sigma\in [101^\omega]$ by Lemma~\ref{lem:101-validity_collapse}. If $\sigma$ starts with 100, $\sigma\in [10^\omega]$ by Lemma~\ref{lem:10k-validity_collapse}.

          \smallskip\noindent\textbf{Disjointness} We separate the displayed equivalence classes into
          \begin{equation*}
            A=\{[0^\omega], \bigsqcup_{k\geq 1}[01^k], [01^\omega], [10], [10^\omega], [101^\omega], [1^\omega]\}
          \end{equation*}
          and
          \begin{equation*}
            B=\{\bigsqcup_{k\geq 2}[0^k]\}.
          \end{equation*}
          For the disjointness in $A$, we can use the same witnesses as \ref{itm:classification_single_kd45}. For the disjointness in $B$, use Lemmas~\ref{lem:0k_but_not_0kplus1_multi_kd45} and \ref{lem:0k-valid_but_not_0_plus1-validity}.

          It remains to compare the classes in $A$ with those in $B$.
          Fix $k\geq 2$.
          First, $[0^k]\neq [0^\omega]$ by Lemma~\ref{lem:0k_but_not_0kplus1_multi_kd45} and \ref{lem:0k-valid_but_not_0_plus1-validity}. Second, $[0^k]\neq [1^\omega]$ since $\Box_i\bot$ is $1^\omega$-valid but not $0^k$-valid. In fact, after the announcement of $\Box_i\bot$, all $i$-arrows are deleted, making $\Box_i\bot$ true. Finally, $p$ separates $[0^k]$ from all the remaining classes in $A$.
  \end{enumerate}
\end{proof}

\section{Lemmas for the classification theorem}\label{sec:Lemmas_for_the_classification_theorem}
In this section, we prove various lemmas for the classification theorem. 
\subsection{Collapse lemmas}
\begin{lemma}[00-validity collapse]\label{lem:00-validity_collapse}

  In multi-agent K45 and single-agent KD45: Every 00-valid formula is $0^\omega$-valid.
\end{lemma}
\begin{proof}
  Let $\varphi$ be any 00-valid formula and $M,w$ be any model. If $M,w\models\lnot\varphi$, we have $M|\varphi,w\models\lnot\varphi$ by 00-validity.

  \smallskip\noindent
  \textbf{Multi-agent K45}
  Since $M|\varphi$ is again a K45 model, we must have $M|^2\varphi,w\models\lnot\varphi$. Continuing this reasoning, we have that $\varphi$ is $0^\omega$-valid.

  \smallskip\noindent
  \textbf{Single-agent KD45}
  If $R|\varphi(w)=\varnothing$, the model never changes so $\varphi$ is $0^\omega$-valid. If $R|\varphi(w)\neq\varnothing$, the generated submodel $M_w$ is again KD45 since $R|\varphi(w)=R|\varphi(v)$ for all $v\in M_w$. Thus, again we have $M|^n\varphi,w\models\lnot\varphi$ for all $n\geq 2$ hence $\varphi$ is $0^\omega$-valid.
\end{proof}

\begin{lemma}[11-validity collapse]\label{lem:11-validity_collapse}

  In multi-agent K45, KD45, and S5: Every 11-valid formula is $1^\omega$-valid.
\end{lemma}
\begin{proof}
  Suppose that $\varphi$ is $11$-valid in multi-agent K45, KD45, or S5 and let $T_n=\llbracket\varphi\rrbracket_{M|^n\varphi}$. Then, $T_0\subseteq T_1$, so for any K45, KD45, or S5 model $M=(W,\{R_i\}_{i\in G},V)$ and $i\in G$, we have
  \begin{equation*}
    R_i|^2\varphi = R_i\cap(W\times T_0)\cap(W\times T_1) = R_i\cap(W\times T_0) = R_i|\varphi.
  \end{equation*}
  Thus, $M|^2\varphi=M|\varphi$. Now, if $M,w\models\varphi$, then $M|\varphi,w\models\varphi$ by 11-validity so the truth of $\varphi$ at $w$ remains 1 forever. Thus $\varphi$ is $1^\omega$-valid.
\end{proof}

\begin{lemma}[$10^k$-validity collapse]\label{lem:10k-validity_collapse}
  In multi-agent K45: For each $k\geq 1$, every $10^k$-valid formula is $10^{k+1}$-valid. In multi-agent KD45 and S5:
  \begin{itemize}
    \item There is a non-trivially 10-valid formula that is not 100-valid.
    \item For $k\geq 2$, every $10^k$-valid formula is $10^{k+1}$-valid.
  \end{itemize}
\end{lemma}
\begin{proof}
  \smallskip\noindent\textbf{Multi-agent K45}

  Suppose that $\varphi$ is $10^k$-valid for some $k\geq 1$. Suppose that $M,w\models\varphi$ for some K45 model $M,w$ and for each $n\geq 0$, let $T_n:=\llbracket \varphi\rrbracket_{M|^n\varphi}$. Since $\varphi$ is $10$-valid, we have $M|\varphi,w\models\lnot\varphi$, so $T_0\cap T_1=\varnothing$. Thus,
  \begin{equation*}
    R_i|^2\varphi=R_i\cap (W\times T_0)\cap (W\times T_1)=R_i\cap (W\times (T_0\cap T_1))=\varnothing.
  \end{equation*}
  This means $M|^2\varphi=M|^3\varphi=\cdots$ by the monotonicity of arrow-elimination, so $M|^n\varphi\models\lnot\varphi$ for all $n\geq 2$. Furthermore, we claim that $M|^2\varphi,w\models\lnot\varphi$; otherwise, $M|^2\varphi,w\models\varphi$ and since $M|^2\varphi$ is again a K45 model (this reasoning fails for KD45), we would have $M|^3\varphi,w\models\lnot\varphi$, contradicting $M|^2\varphi=M|^3\varphi$. Thus, $\varphi$ is $10^\omega$-valid.

  \smallskip\noindent\textbf{Multi-agent KD45 and S5} The proof that $10^k$-validity implies $10^{k+1}$-validity for $k\geq 2$ is the same as above except that we use 100-validity to claim $M|^2\varphi,w\models\lnot\varphi$.

  For  the existence of a formula that is non-trivially 10-valid but not 100-valid, we show that $\varphi:= (p\wedge\Diamond_a p\wedge\Diamond_a\neg p) \vee \Box_a\bot$ is such an example.\footnote{$\varphi$ is actually non-trivially $101$-valid. In fact, for $101$-validity, take any $u\in R_a|\varphi(w)$. Then, we have $M|\varphi,u\not\models\Diamond_a\lnot p$ and $M|\varphi,u\not\models\Box_a\bot$ since $M|\varphi$ is K45, so $M|\varphi,u\not\models\varphi$. Thus, $M|^2\varphi,w\models\Box_a\bot$ hence $M|^2\varphi,w\models\varphi$. For $101$-satisfiability, use the same model in the proof. This is exactly the statement of Lemma~\ref{lem:101-validity} in a later subsection.}

  To show that $\varphi$ is 10-valid, let $M,w$ be any KD45 model and suppose that $M,w\models\varphi$. Then $M,w\models p\land\Diamond_a p\land\Diamond_a\lnot p$ since $\Box_a\bot$ is false. Take any $u\in W$ with $wR_a u$. Then, $M,u\models\Diamond_a p\land\Diamond_a\lnot p$ by Lemma~\ref{lem:modal agreement lemma}, so $u\in M|\varphi$ iff $M,u\models\varphi$ iff $M,u\models p$. Thus, $M|\varphi,w\models\Box_a p$ so $M|\varphi,w\not\models\Diamond_a\lnot p$. Also, $M|\varphi,w\not\models\Box_a\bot$ by $M,w\models\Diamond_a p$ and $M,u\models\Diamond_a p\land\Diamond_a \lnot p$ for all $u\in W$ with $wR_au$. Thus, we have $M|\varphi,w\not\models\varphi$ hence $\varphi$ is 10-valid.

  To show that $\varphi$ is not 100-valid, define $M=(W, \{R_i\}_{i\in G}, V)$ by $W=\{w,u,v\}$, $R_i=W\times W$ for all $i\in G$, and $V(p)=\{w,u\}$.

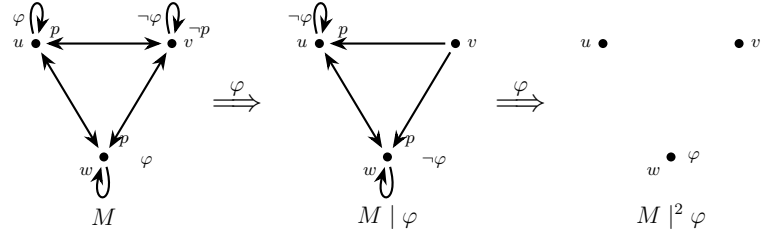
  \begin{figure}[htbp]\label{fig: 10_k-validity_collapse}
    \centering
    \begin{tikzpicture}[
        >=Stealth,
        world/.style={circle, fill=black, inner sep=1.6pt},
        acc/.style={->, thick, shorten >=2pt, shorten <=2pt},
        lab/.style={font=\small},
        every loop/.style={looseness=7, min distance=8mm},scale=0.75, transform shape
      ]

      \node[world,label={[lab]left:$u$}] (u0) at (0,2) {};
      \node[world,label={[lab]right:$v$}] (v0) at (2.4,2) {};
      \node[world,label={[lab]below left:$w$}] (w0) at (1.2,0) {};

      \path[acc]
      (u0) edge[loop above] (u0)
      (v0) edge[loop above] (v0)
      (w0) edge[loop below] (w0);

      \draw[<->, thick, shorten >=2pt, shorten <=2pt] (u0) -- (v0);
      \draw[<->, thick, shorten >=2pt, shorten <=2pt] (u0) -- (w0);
      \draw[<->, thick, shorten >=2pt, shorten <=2pt] (v0) -- (w0);

      \node[lab] at (-0.30,2.42) {$\varphi$};
      \node[lab] at (0.34,2.20) {$p$};

      \node[lab] at (2.00,2.42) {$\neg\varphi$};
      \node[lab] at (2.88,2.20) {$\neg p$};

      \node[lab] at (1.56,0.28) {$p$};
      \node[lab] at (1.95,-0.08) {$\varphi$};

      \node[font=\large] at (1.2,-1.05) {$M$};

      \node[font=\Large] at (3.55,1.1) {$\overset{\varphi}{\Longrightarrow}$};

      \node[world,label={[lab]left:$u$}] (u1) at (5.0,2) {};
      \node[world,label={[lab]right:$v$}] (v1) at (7.4,2) {};
      \node[world,label={[lab]below left:$w$}] (w1) at (6.2,0) {};

      \path[acc]
      (u1) edge[loop above] (u1)
      (w1) edge[loop below] (w1)
      (v1) edge (u1)
      (v1) edge (w1);

      \draw[<->, thick, shorten >=2pt, shorten <=2pt] (u1) -- (w1);

      \node[lab] at (4.66,2.42) {$\neg\varphi$};
      \node[lab] at (5.42,2.20) {$p$};

      \node[lab] at (6.60,0.28) {$p$};
      \node[lab] at (7.02,-0.08) {$\neg\varphi$};

      \node[font=\large] at (6.2,-1.05) {$M\mid\varphi$};

      \node[font=\Large] at (8.55,1.1) {$\overset{\varphi}{\Longrightarrow}$};

      \node[world,label={[lab]left:$u$}] (u2) at (10.0,2) {};
      \node[world,label={[lab]right:$v$}] (v2) at (12.4,2) {};
      \node[world,label={[lab]below left:$w$}] (w2) at (11.2,0) {};
      \node[lab] at (11.6,0) {$\varphi$};
      \node[font=\large] at (11.2,-1.05) {$M\mid^{2}\varphi$};

    \end{tikzpicture}
  \end{figure}

  Then, we have $M,w\models\varphi$, $M|\varphi,w\models\lnot\varphi$ but $M|^2\varphi,w\models\varphi$, as can be seen from the figure. Thus, $\varphi$ is non-trivially 10-valid but not 100-valid.
\end{proof}

\begin{lemma}[101-validity collapse]\label{lem:101-validity_collapse}
  In multi-agent KD45 and S5: Every 101-valid formula is $101^\omega$-valid.
\end{lemma}
\begin{proof}
  If $\varphi$ is $101$-valid, then $T_0\cap T_1=\varnothing$. Hence, for all $i\in G$, $R_i|^2\varphi = R_i\cap(W\times T_0)\cap(W\times T_1) = \varnothing$, so we have $M|^2\varphi=M|^3\varphi=\cdots$. Thus, $\varphi$ is $101^\omega$-valid.
\end{proof}

\subsection{Nonexistence lemmas}
\subsubsection{Nonexistence in multi-agent K45, single-agent KD45, and multi-agent S5}
\begin{lemma}[Nonexistence of non-trivially $0^k1$-validity]\label{lem:nonexistence_of_0k1-validity}
  In multi-agent S5: For all $k\geq 2$, there is no non-trivially $0^k1$-valid formula.
\end{lemma}
\begin{proof}
  Suppose toward a contradiction that there is a non-trivially $0^k1$-valid formula $\varphi$ for some $k\geq 2$. Then there is a finite S5 model $M,w$ such that $M,w\models\lnot\varphi,M|\varphi,w\models\lnot\varphi,\ldots,M|^{k-1}\varphi,w\models\lnot\varphi,M|^k\varphi,w\models\varphi$. Let $A_0=W$ and $A_k=\bigcap_{j<k}T_j$ for $k\geq 1$ where $T_j=\llbracket\varphi\rrbracket_{M|^j\varphi}$. Note that $\{A_k\}_{k=0}^\infty$ is a decreasing sequence and $R_i|^k\varphi=R_i\cap (W\times A_k)$ holds. Since the truth of $\varphi$ at $w$ changes when shifting from $M|^{k-1}\varphi$ to $M|^k\varphi$, we must have $A_{k-1}\supsetneq A_k$.

  Take a $u\in A_{k-1}\backslash A_k$. Then, $M|^{k-1}\varphi,u\models\lnot\varphi$ by $u\notin T_{k-1}$. Now consider the submodel $M\restriction A_{k-1}$. Since $M$ is S5, $M\restriction A_{k-1}$ is also S5, so we have 
  \begin{align*}
    (M\restriction A_{k-1}),u&\models\lnot\varphi, \,(M\restriction A_{k-1})|\varphi,u\models\lnot\varphi, \ldots,\, (M\restriction A_{k-1})|^{k-1}\varphi,u\models\lnot\varphi,\\ (M\restriction A_{k-1})|^k\varphi,u&\models\varphi 
  \end{align*}
    again by $0^k1$-validity of $\varphi$. Note that in general, the accessibility relations of $M\restriction A_n$ and $M|^n\varphi$ for each $n\geq 1$ are given by $R_i\cap (A_n\times A_n)$ and $R_i\cap(W\times A_n)$, respectively, so that \begin{equation*}
      (M\restriction A_n)|^m,v\models\psi\iff M|^{n+m}\varphi,v\models\psi
    \end{equation*}
     holds for all formulas $\psi$, $m\geq 0$, and $v\in A_n$. Thus, from $(M\restriction A_{k-1})|^{k-1}\varphi,u\models\lnot\varphi$ and $(M\restriction A_{k-1})|^k\varphi,u\models\varphi$, we get $M|^{2k-2}\varphi,u\models\lnot\varphi$ and $M|^{2k-1}\varphi,u\models\varphi$. This implies $A_{2k-2}\supsetneq A_{2k-1}$ so we can choose a $u_1\in A_{2k-2}\backslash A_{2k-1}$.

  Repeating this reasoning gives infinitely many strict decreases, contradicting that $M$ is finite. Therefore, there is no non-trivially $0^k1$-valid formula for all $k\geq 2$.
\end{proof}

\begin{lemma}[Nonexistence of non-trivially $01^k0$-validity]\label{lem:nonexistence_of_01k0-validity}
  In multi-agent K45, single-agent KD45, and multi-agent S5: For each $k\geq 1$, there is no non-trivially $01^k0$-valid formula.
\end{lemma}
\begin{proof}
  If there were a non trivially $01^k0$-valid formula $\varphi$, there would be a finite K45, KD45, or S5 model $M,w$ such that $M,w\models\lnot\varphi$ by the $01^k0$-satisfiability of $\varphi$ and the finite model property.

  \smallskip\noindent\textbf{Multi-agent K45 and single-agent KD45} If $M$ is K45, $M|^n\varphi$ is also K45 for all $n\geq 1$. We show that this also holds for single-agent KD45. Suppose that $M$ is single-agent KD45. Suppose for contradiction that $M|\varphi$ is not serial. Then, $R|\varphi(w)=\varnothing$ for some $w\in M|\varphi$. Since $M$ is serial, choose $v$ with $wRv$. Then, by $R(w)=R(v)$, we have $R|\varphi(v)=R|\varphi(w)=\varnothing$. Moreover, $M,v\models\lnot\varphi$ since $wRv$ and $R|\varphi(w)=\varnothing$, so we have $M|\varphi,v\models\varphi$ by $01^k0$-validity of $\varphi$. However, $R|\varphi(v)=\varnothing$ implies that the truth of $\varphi$ remains true after that (this part does not apply to multi-agent KD45), contradicting $01^k0$-validity. Thus, $M|\varphi$ is also single-agent KD45, and more generally, $M|^n\varphi$ is single-agent KD45 for all $n\geq 1$.

  Therefore, both in multi-agent K45 and single-agent KD45, we would have the truth pattern $0(1^k0)^\omega$, contradicting the fact that the truth of $\varphi$ at $w$ must stabilize since $M$ is finite.

  \smallskip\noindent\textbf{Multi-agent S5} Suppose that $M$ is S5. By $01^k0$-validity, we have $M,w\models\lnot\varphi,M|\varphi,w\models\varphi,\ldots,M|^k\varphi,w\models\varphi, M|^{k+1}\varphi,w\models\lnot\varphi$. Let $A_0=W$ and $A_k=\bigcap_{j<k}T_j$ for $k\geq 1$ where $T_j=\llbracket\varphi\rrbracket_{M|^j\varphi}$. Then, $R_i|^k\varphi=R_i\cap (W\times A_k)$. Since the truth of $\varphi$ at $w$ changes when shifting from $M|^{k}\varphi$ to $M|^{k+1}\varphi$, we must have $A_{k}\supsetneq A_{k+1}$. Take any $u\in A_{k}\backslash A_{k+1}$. Then, $M|^{k}\varphi,u\models\lnot\varphi$ hence $M\restriction A_{k},u\models\lnot\varphi$. Applying $01^k0$-validity to the S5 model $M\restriction A_{k}$ yields $A_{2k}\supsetneq A_{2k+1}$ by the same argument as Lemma~\ref{lem:nonexistence_of_0k1-validity}. Repeating this contradicts the fact that $M$ is finite. Thus, there is no non-trivially $01^k0$-valid formula.
\end{proof}

\subsubsection{Nonexistence in multi-agent KD45 with $|G|\geq 2$}
\begin{definition}\label{def:m_w_and_n_r_unravelling}
  Let $\varphi$ be a formula with $d:=md(\varphi)\geq 1$, and let $I$ be the set of agents occurring in $\varphi$. Suppose that a KD45 model $M=(W,\{R_i\}_{i\in I}, V)$ and $w\in M$ satisfies $M,w\models\lnot\varphi$. From $M,w$, We construct a KD45 model $N=(W^N, \{R_i^N\}_{i\in I}, V^N)$ and $r\in N$ as follows (see Figure~\ref{fig:example_of_mw_and_nr} for an example).
\begin{itemize}
  \item We define $W^N$ as follows. First by recursion on $1\leq h\leq d$, define the set $\mathcal{A}_h$, the family of sets $\{D_{(x,i)}\}_{(x,i)\in \mathcal{A}_h}$, and a function $\ell$ as follows.
  
  \smallskip\noindent\textbf{
  Base case} Take a new state $r\notin W$ and let $\ell(r):=w$. Next, let $\mathcal{A}_1:=\{(r,i)\colon i\in I\}$. For each $(r,i)\in\mathcal{A}_1$, let $D_{(r,i)}:=\{((r,i),v)\colon v\in R_i(w)\}$ and define $\ell\colon D_{(r,i)}\to W$ by $\ell((r,i),v):=v$.
  
  \smallskip\noindent\textbf{Recursive case} Assume that $\mathcal{A}_h$ and $D_{(x,i)}$ are all defined for $1\leq h< d$ and $(x,i)\in \mathcal{A}_h$. Then, let 
  \begin{equation*}
  \mathcal{A}_{h+1}:=\{(y,j)\colon \exists (x,i)\in \mathcal{A}_h[y\in D_{(x,i)}\text{ and }j\neq i]\}.
  \end{equation*}
  For each $(x,i)\in \mathcal{A}_{h+1}$, let
  \begin{equation*}
    D_{(x,i)}:=\{((x,i),v)\colon v\in R_i(\ell(x))\}
  \end{equation*}
   and define $\ell\colon D_{(x,i)}\to W$ by $\ell((x,i),v):=v$.

   Finally, let $W_0^N:=\{r\}$ and $W_h^N:=\bigsqcup_{(x,i)\in \mathcal{A}_h}D_{(x,i)}\,(h\geq 1)$. Define the set of states by
   \begin{equation*}
    W^N=\bigsqcup_{h=0}^d W_h^N.
   \end{equation*}
   So, the model $N$ consists of the root $r$ and the states in the cells $D_{(x,i)}$ for all $1\leq h\leq d$ and $(x,i)\in\mathcal{A}_h$. The function $\ell\colon W^N\to W$ assigns to $((x,i),v)\in W^N$ the corresponding state $v\in W$ in the original model $M$. $W_N$ is indeed finite since $M$ is finite. 
   \item Define the \emph{rank} of $x\in W^N$ by $\text{rank}(x)=h\iff x\in W^N_h$ for all $0\leq h\leq d$. For each $x\in W^N\backslash\{r\}$, there is a unique $1\leq h\leq d$ and $(y,j)\in\mathcal{A}_h$ such that $x\in D_{(y,j)}$. Using this unique $(y,j)$, for each $i\in I$, define the accessibility relation by
   \begin{equation}\label{eq:accessibility_relation_of_n}
    R_i^N(x):=\begin{cases}
      D_{(r,i)} & x=r\\
      D_{(y,j)} & x\neq r,\, i=j,\, \text{rank}(x)<d\\
      D_{(x,i)} & x\neq r,\, i\neq j,\, \text{rank}(x)< d\\
      D_{(y,j)} &\text{rank}(x)=d
    \end{cases}
   \end{equation}
   So, at the root $r$, there is an $i$-arrow from $r$ to the states in the cell $D_{(r,i)}$. If $\text{rank}(x)<d$ and $i=j$, $R_i^N(x)$ forms a cell within $D(y,j)$. If $\text{rank}(x)<d$ and $i\neq j$, there is an $i$-arrow from $x$ to the states in the cell $D_{(x,i)}$. If $\text{rank}(x)=d$, $R_i^N(x)$ forms a cell within $D_{(y,j)}$. Note that $R_i^N$ is KD45. In fact, since $R_i$ is serial, the right hand side of (\ref{eq:accessibility_relation_of_n}) is nonempty in all the cases, so $R_i^N$ is serial. Also, it is easy to see from the definition of $R_i^N$ that $xR_i^N y$ implies $R_i^N(x)=R_i^N(y)$ for all $x,y\in W^N$, so $R_i^N$ is transitive and Euclidean.
   \item Define the valuation function by
   \begin{equation*}
    V^N(p)=\ell^{-1}[V(p)]
   \end{equation*}
   for all $p\in\mathbf{Prop}$.
   That is, the valuations are derived from the original states.
  \end{itemize}
\end{definition}

  \begin{figure}[htbp]
    \centering
    \begin{tikzpicture}[
  >=Stealth,
  thick,
  dot/.style={circle, fill=black, inner sep=1.4pt},
  core/.style={draw, rounded corners, inner sep=4pt},
  tuple/.style={font=\tiny, inner sep=1pt},
  every node/.style={font=\small}
]

\node[dot,label=above:$w$] (w) at (0,4.2) {};
\node[dot,label=below:$w_1$] (w1) at (-1.6,2.4) {};
\node[dot,label=above:$w_2$] (w2) at (1.4,2.4) {};
\node[dot,label=below:$w_3$] (w3) at (1.4,0.8) {};

\draw[->] (w) -- node[left] {$a,b$} (w1);
\draw[->] (w) -- node[right] {$b$} (w2);
\draw[<->] (w1) -- node[above] {$b$} (w2);
\draw[->] (w2) -- node[right] {$a$} (w3);

\draw[->] (w1) edge[loop left] node[left] {$a,b$} ();
\draw[->] (w2) edge[loop right] node[right] {$b$} ();
\draw[->] (w3) edge[loop right] node[right] {$a,b$} ();

\node at (0,0) {$M,w$};

\begin{scope}[xshift=-10cm,yshift=-6cm]
\node[dot,label=above:$r$] (r) at (10.0,4.5) {};

\node[dot] (raw1) at (7.0,2.9) {};
\node[tuple,above left=3pt and 0pt of raw1] {$(r,a,w_1)$};

\node[dot] (rbw1) at (10.0,2.9) {};
\node[tuple,above left=3pt of rbw1] {$(r,b,w_1)$};

\node[dot] (rbw2) at (13.0,2.9) {};
\node[tuple,above right=2pt and 0pt of rbw2] {$(r,b,w_2)$};

\draw[->] (r) -- node[above left] {$a$} (raw1);
\draw[->] (r) -- node[left] {$b$} (rbw1);
\draw[->] (r) -- node[above right] {$b$} (rbw2);

\draw[->] (raw1) edge[loop left] node[left] {$a$} ();
\draw[->] (rbw1) edge[loop left] node[left] {$b$} ();
\draw[->] (rbw2) edge[loop right] node[right] {$b$} ();
\draw[<->] (rbw1) -- node[above] {$b$} (rbw2);

\node[dot] (rabw1) at (6.0,1.0) {};
\node[tuple,below left=3pt and 0pt of rabw1] {$(r,a,w_1,b,w_1)$};

\node[dot] (rabw2) at (8.0,1.0) {};
\node[tuple,below=4pt and 0pt of rabw2] {$(r,a,w_1,b,w_2)$};

\node[dot] (rbaw1) at (10.0,1.0) {};
\node[tuple,below=2pt of rbaw1] {$(r,b,w_1,a,w_1)$};

\node[dot] (rbaw3) at (13.0,1.0) {};
\node[tuple,below=2pt of rbaw3] {$(r,b,w_2,a,w_3)$};

\draw[->] (raw1) -- node[left] {$b$} (rabw1);
\draw[->] (raw1) -- node[right] {$b$} (rabw2);
\draw[->] (rbw1) -- node[left] {$a$} (rbaw1);
\draw[->] (rbw2) -- node[right] {$a$} (rbaw3);

\draw[->] (rabw1) edge[loop left] node[left] {$a,b$} ();
\draw[->] (rabw2) edge[loop right] node[right] {$a,b$} ();
\draw[<->] (rabw1) -- node[above=3pt] {$a,b$} (rabw2);

\draw[->] (rbaw1) edge[loop right] node[right] {$a,b$} ();
\draw[->] (rbaw3) edge[loop right] node[right] {$a,b$} ();

\node[core,fit=(raw1),label={right=:$D_{(r,a)}$}] {};
\node[core,fit=(rbw1)(rbw2),label={[yshift=-0.3em]below:$D_{(r,b)}$}] {};
\node[core,fit=(rabw1)(rabw2),label={[yshift=-0.3em]below:$D_{(r,a,w_1,b)}$}] {};
\node[core,fit=(rbaw1),label={[yshift=-0.3em]below:$D_{(r,b,w_1,a)}$}] {};
\node[core,fit=(rbaw3),label={[yshift=-0.3em]below:$D_{(r,b,w_2,a)}$}] {};

\node[right] at (15.3,4.5) {rank 0};
\node[right] at (15.3,2.9) {rank 1};
\node[right] at (15.3,1.0) {rank 2};

\node at (10.0,0) {$N,r$};
\end{scope}
\end{tikzpicture}
\caption{An example of $M,w$ and $N,r$ in Definition~\ref{def:m_w_and_n_r_unravelling}. Here, $M,w$ is a KD45 model with $I=\{a,b\}$ and $d=2$ intended. $N,r$ is the corresponding KD45 unravelling.}
\label{fig:example_of_mw_and_nr}
  \end{figure}
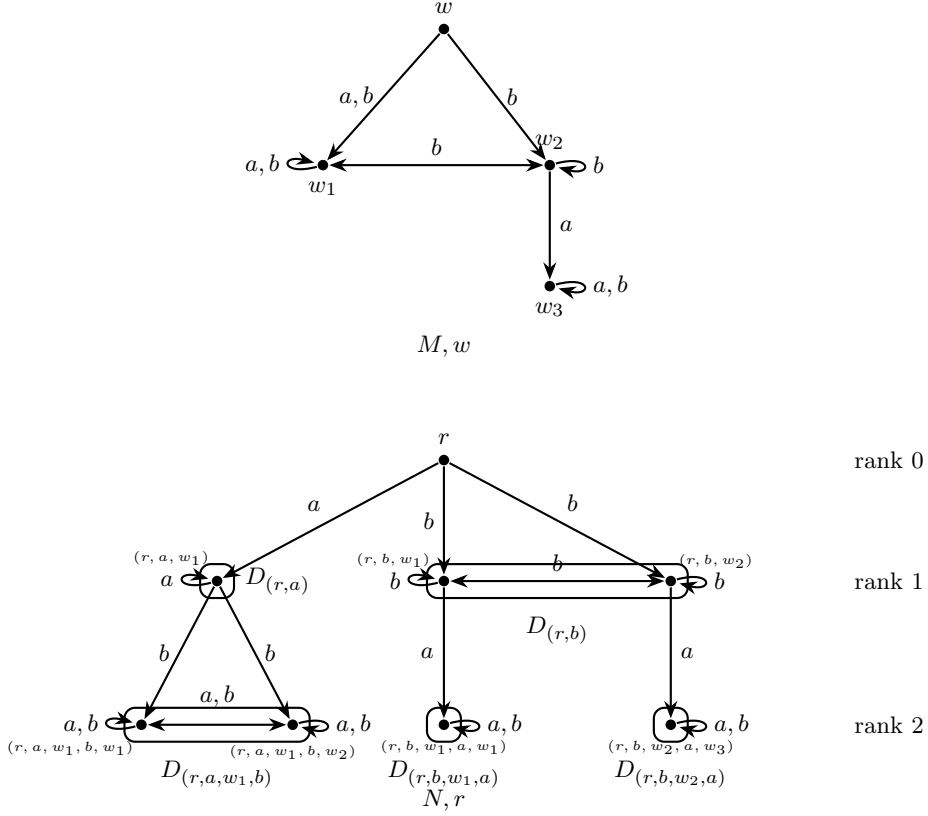

\begin{lemma}
    Let $\varphi$, $d$, $M,w$, and $N,r$ as in Definition~\ref{def:m_w_and_n_r_unravelling}. Then, for all $0\leq h\leq d$, $x\in W^N_h$, and formulas $\psi$ with $md(\psi)\leq d-h$, we have
    \begin{equation*}
    N,x\models\psi\iff M,\ell(x)\models\psi
    \end{equation*}
     In particular, $N,r\models \varphi\iff M,\ell(r)\models \varphi$.
  
  \end{lemma} 
  \begin{proof}
    We prove the claim by induction on $\psi$.
    
    First note that we have the equation
    \begin{equation}\label{eq:accessibility_relation_correspondence_between_m_n}
      \ell[R_i^N(x)]=R_i(\ell(x))
    \end{equation}
    for all $i\in I$, $0\leq h<d$ and $x\in W^N_h$. Moreover, whenever $xR_i^Nz$, we have
    \begin{equation}\label{eq:rank_z_leq_rank_x_plus1}
      \text{rank}(z)\leq\text{rank}(x)+1.
    \end{equation}
     In fact, if $x=r$, a move by $i$ from $x$ increases the rank of $x$ by one. If $x\neq r$, a move by $i$ does not change the rank of $x$ whereas a move by $j\neq i$ increases the rank by at most one.

     Now, we prove the equivalence. If $\psi=p$, we have $N,x\models\psi$ iff $x\in V^N(p)$ iff $\ell(x)\in V(p)$ iff $M,\ell(x)\models \psi$ by the definition of $N$. The Boolean cases are immediate. 
     
     It remains to consider the $\psi=\Box_i\chi$ case. Since $md(\Box_i\chi)=md(\psi)\leq d-h$, we have $md(\chi)\leq d-h-1$.

     First suppose that $N,x\models\Box_i\chi$. Let $v\in R_i(\ell(x))$. By (\ref{eq:accessibility_relation_correspondence_between_m_n}), there is some $z\in R_i^N(x)$ such that $\ell(z)=v$. Thus, we have $md(\chi)\leq d-(h+1)\leq d-\text{rank}(z)$ by (\ref{eq:rank_z_leq_rank_x_plus1}) and $\text{rank}(x)=h$. Since $N,z\models\chi$, applying the inductive hypothesis to $\chi$ and $z$, yields $M,v=M,\ell(z)\models\chi$. As $v\in R_i(\ell(x))$ was arbitrary, we obtain $M,\ell(x)\models\Box_i\chi$.

     Conversely, suppose that $M,\ell(x)\models\Box_i\chi$. Let $z\in R_i^N(x)$. By (\ref{eq:accessibility_relation_correspondence_between_m_n}), $\ell(z)\in R_i(\ell(x))$ and hence $M,\ell(z)\models\chi$. Again, (\ref{eq:rank_z_leq_rank_x_plus1}) implies $md(\chi)\leq d-\text{rank}(z)$, so the inductive hypothesis gives $N,z\models\chi$. Since $z\in R_i^N(x)$ was arbitrary, $N,x\models\Box_i\chi$.

     This completes the induction.  
  \end{proof}

For a formula $\varphi$, let $\varphi^\Delta$ be the formula obtained from $\varphi$ by replacing $\Box_i$ for all $i\in I$ with the single modality $\Box$. The following lemma is used to show that $\text{rank}(x)<d$ for the specific $x$ in the proof of Lemmas~\ref{lem:nonexistence_of_0k1-validity_multi_kd45} and \ref{lem:nonexistence_of_01k0-validity_multi_kd45}.
\begin{lemma}\label{lem:diagonal_model_equivalence}
  Let $\varphi$, $d$, $M,w$, and $N,r$ as in Definition~\ref{def:m_w_and_n_r_unravelling}.
\begin{enumerate} 
  \item Let $P=(W^P,R^P,V^P)$ be any single-agent KD45 model and define the multi-agent KD45 model $P^I=(W^P,\{R_i^P\}_{i\in I}, V^P)$ by $R_i^P=R^P$ for all $i\in I$. Then,
  \begin{equation*}
    P,u\models\psi^\Delta\iff P^I,u\models\psi
  \end{equation*}
  for all formulas $\psi$.
  \item Suppose that $\varphi^\Delta$ is valid in single-agent KD45. Then, for all $x\in N$ with $\operatorname{rank}(x)=d$ and all $n\geq 0$, we have $N|^n\varphi,x\models\varphi$.  
\end{enumerate}  
\end{lemma}
\begin{proof}
(1) The statement follows from induction on $\psi$. 

(2) Take any $x\in N$ with $\text{rank}(x)=d$, and let $D_{(y,j)}$ be the unique cell such that $x\in D_{(y,j)}$.
Let $N^\Delta_{D_{(y,j)}}$ be the single-agent model whose domain is $D_{(y,j)}$, whose accessibility relation is $D_{(y,j)}\times D_{(y,j)}$, and whose valuation is the restriction of $V^N$ to $D_{(y,j)}$. By the definition of $N$ and $\text{rank}(x)=d$, we have $R_i^N(z)=D_{(y,j)}$ for all $z\in D_{(y,j)}$ and $i\in I$. Hence, by (1) we have $N,z\models\psi$ iff $N^\Delta_{D_{(y,j)}},z\models\psi^\Delta$ for all $z\in D_{(y,j)}$ and formulas $\psi$.

Since $N^\Delta_{D_{(y,j)}}$ is a single-agent KD45 model and $\varphi^\Delta$ is valid in single-agent KD45, we have $N^\Delta_{D_{(y,j)}},z\models\varphi^\Delta$ for all $z\in D_{(y,j)}$. Therefore, $N,z\models\varphi$ for all $z\in D_{(y,j)}$.

It follows that the first update by $\varphi$ removes no arrow whose source lies in $D_{(y,j)}$, so the pointed model generated by every $z\in D_{(y,j)}$ remains unchanged after the update. Repeating the same argument, we have $N|^n\varphi,z\models\varphi$ for all $n\geq 0$ and $z\in D_{(y,j)}$. In particular, $N|^n\varphi,x\models\varphi$ for every $n\geq 0$.
\end{proof}
For each $n\geq 0$, let $N_n:=N|^n\varphi$. For each $y\in N$, let
\begin{equation*}
  f(y):=\min\{n\geq 0\colon N_n,y\models\lnot\varphi\}
\end{equation*}
be the the first time point at which $\varphi$ becomes false at $y$. If no such time point exits, we let $f(y)=\infty$. 
\begin{lemma}[Nonexistence of non-trivially $0^k1$-validity in KD45]
\label{lem:nonexistence_of_0k1-validity_multi_kd45}
In multi-agent KD45: For each $k\geq 2$, there is no non-trivially
$0^k1$-valid formula.
\end{lemma}

\begin{proof}
Fix $k\geq 2$. It is enough to prove the stronger statement that if $\varphi$ is $0^k1$-valid, then $\varphi$ is valid since validity contradicts $0^k1$-satisfiability. Suppose that $\varphi$ is $0^k1$-valid and assume toward a contradiction that $\varphi$ is not valid. Furthermore, let $I$ be the set of agents that appear in $\varphi$. Hereafter, we work with the set $I$ since the truth of $\varphi$ does not depend on the agents outside $I$. By the finite model property, there is a finite KD45 model $M=(W,\{R_i\}_{i\in I}, V)$ and $w\in M$ such that $M,w\models\lnot\varphi$. To derive a contradiction, consider the model $N,r$ obtained from $M,w$ as in Definition~\ref{def:m_w_and_n_r_unravelling}. Choose a state $x$ with maximal rank among the states $y$ such that $f(y)<\infty$. Let $q=f(x)$. Also, let $N_n=N|^n\varphi$ for $n\geq 0$, and let $A_0=W^N$ and $A_n=\bigcap_{\ell=0}^{n-1}\llbracket \varphi\rrbracket_{N_\ell}$ for $n\geq 1$.

First note that $x$ cannot be $r$. In fact, if $x=r$, every successor of $r$ other than $r$ would satisfy $\varphi$ forever by maximality of $x$ hence the model $N$ would never change after the first announcement, contradicting $0^k1$-validity. 
We next show that $\text{rank}(x)<d$. Assume toward a contradiction that $\text{rank}(x)=d$. Since $\varphi$ is $0^k1$-valid,  $\varphi^\Delta$ is also $0^k1$-valid by Lemma~\ref{lem:diagonal_model_equivalence} (the truth equivalence is preserved through updates). In particular, the single-agent formula $\varphi^\Delta$ is $00$-valid by $k\geq 2$ hence $0^\omega$-valid by Lemma~\ref{lem:00-validity_collapse}. Thus, $\varphi^\Delta$ must be valid. Therefore, $N_n,x\models\varphi$ for all $n\geq 0$ by Lemma~\ref{lem:diagonal_model_equivalence}. However, this contradicts $f(x)<\infty$ so that we have $\text{rank}(x)<d$.

Let $D_{(y,j)}$ be such that $x\in D_{(y,j)}$. 
Now, since $q$ is the first time at which $\varphi$ is false at $x$, we have $x\in A_q$. Thus, we get $D_{(y,j)}\cap A_q\neq\varnothing$ hence $R_j^{N_q}(x)=D_{(y,j)}\cap A_q\neq \varnothing$, so $R_j^{(N_q)_x}$ is serial. For $i\neq j$, $R_i(x)=D_{(x,i)}$ by definition, and since all the states in $D_{(x,i)}$ has rank greater than $\text{rank}(x)$, they all satisfy $\varphi$ forever. Thus, $R_i^{N_q}(x)=D_{(x,i)}\neq\varnothing$ hence $R_i^{(N_q)_x}$ is serial. Furthermore, transitivity and Euclideanness are preserved through updates. Thus, the generated submodel $(N_q)_x$ is again KD45.

Apply the $0^k1$-validity to $(N_q)_x$. Since $N_q,x\models\lnot\varphi$, we have $N_{q+m},x\models\lnot\varphi$ for all $0\leq m\leq k-1$ and $N_{q+k},x\models\varphi$. There are two cases.

Suppose that $D_{(y,j)}\cap A_{q+1}\neq\varnothing$. Then, $(N_{q+1})_x$ is again KD45. Since $N_{q+1},x\models\lnot\varphi$, applying $0^k1$-validity to $N_{q+1},x$ yields $N_{q+k},x\models\lnot\varphi$, a contradiction.

Suppose that $D_{(y,j)}\cap A_{q+1}=\varnothing$. Then, the $j$-successor set of $x$ is empty from time $q+1$ onwards. Also, every $i$-successor $(i\neq j)$ of $x$ satisfies $\varphi$ by the maximality of $x$. Thus, the generated submodel at $x$ remains unchanged from $N_{q+1}$ onwards. In particular, the truth value of $\varphi$ at $x$ stays false from time $q+1$ onwards. This contradicts $N_{q+k},x\models\varphi$, which follows from $0^k 1$-validity.

Therefore, for each $k\geq 2$, there is no non-trivially $0^k1$-valid formula.
\end{proof}

\begin{lemma}[Nonexistence of non-trivially $01^k0$-validity in KD45]
\label{lem:nonexistence_of_01k0-validity_multi_kd45}
In multi-agent KD45: for each $k\geq 1$, there is no
non-trivially $01^k0$-valid formula.
\end{lemma}

\begin{proof}
  Fix $k\geq 1$. It is enough to prove the stronger result that if $\varphi$ is $01^k0$-valid, then $\varphi$ is valid. Suppose that $\varphi$ is $01^k0$-valid and assume toward a contradiction that $\varphi$ is not valid. Furthermore, let $I$ be the set of agents that appear in $\varphi$.  By the finite model property, there is a finite KD45 model $M=(W,\{R_i\}_{i\in I}, V)$ and $w\in M$ such that $M,w\models\lnot\varphi$. To derive a contradiction, consider the finite KD45 model $N,r$ obtained from $M,w$. 
  
  Define the lexicographic order $\leq^*$ by
  \begin{align*}
  (\text{rank}(x),f(x))\leq ^* (\text{rank}(y),f(y))\iff
   \text{rank}(x)<\text{rank}(y) \lor(\text{rank}(x)=\text{rank}(y)\land f(x)\leq f(y))
  \end{align*}
  for all $x,y\in N$ with $f(x),f(y)<\infty$. 
  Let $x\in N$ be such that $(\text{rank}(x), f(x))$ is maximal with respect to $\leq ^*$. This exists because $N$ is finite. Let $q=f(x)$. Then, $N_q,x\models\lnot\varphi$. Also, let $N_n=N|^n\varphi$ for $n\geq 0$, and let $A_0=W^N$ and $A_n=\bigcap_{\ell=0}^{n-1}\llbracket \varphi\rrbracket_{N_\ell}$ for $n\geq 1$. 

  Note that $x$ cannot be $r$. In fact, if $x=r$, we have $N,r\models\lnot\varphi$ by $M,w\models\lnot\varphi$, so $N_1,r\models\varphi$ by $01^k0$-validity. However, the truth of $\varphi$ at $r$ cannot change through updates since all the states other than $r$ satisfy $\varphi$ by the maximality of $x$. Thus, $x\neq r$. 

  Furthermore, $\text{rank}(x)<d$. In fact, if $\text{rank}(x)=d$, the single-agent formula $\varphi^\Delta$ must be $01^k0$-valid by $01^k0$-validity of $\varphi$ and Lemma~\ref{lem:diagonal_model_equivalence}. Thus, $\varphi^\Delta$ must be valid since there is no non-trivially $01^k0$-valid formula in single-agent KD45 by Lemma~\ref{lem:nonexistence_of_01k0-validity}. Then, Lemma~\ref{lem:diagonal_model_equivalence} implies that $N|^n,x\models\varphi$ for all $n\geq 0$, contradicting $f(x)<\infty$.

  Let $D_{(y,j)}$ such that $x\in D_{(y,j)}$. By the same reasoning as Lemma~\ref{lem:nonexistence_of_0k1-validity_multi_kd45}, $(N_q)_x$ is again KD45. Since $\varphi$ is $01^k0$-valid, we have $N_{q+m},x\models\varphi$ for all $1\leq m\leq k$ and $N_{q+k+1},x\models\lnot\varphi$. In particular, $N_{q+k},x\models\varphi$ and $N_{q+k+1},x\models\lnot\varphi$. Since $R_i^{N}(x)$ for $i\neq j$ never changes from $q+k$ onwards, we have $D_{(y,j)}\cap A_{q+k}\supsetneq D_{(y,j)}\cap A_{q+k+1}$. Take $z\in D_{(y,j)}\cap A_{q+k}\backslash D_{(y,j)}\cap A_{q+k+1}$. Then, $f(z)=q+k>q=f(x)$ but this contradicts the maximality of $x$ since $x$ and $z$ have the same rank.

  Therefore, for each $k\geq 1$, there is no non-trivially $01^k0$-valid formula.
\end{proof}

\subsection{Existence lemmas}
\begin{lemma}[Existence of non-trivially $0^k$-valid but not $0^{k+1}$-valid formula in KD45]\label{lem:0k_but_not_0kplus1_multi_kd45}
  In multi-agent KD45: There is a formula that is non-trivially $0^k$-valid but not $0^{k+1}$-valid for all $k\geq 2$.
\end{lemma}
\begin{proof}
  Fix $k\geq 2$ and let $P_0,P_1,\ldots,P_{k-1}$ be proposition letters. Let
  \begin{equation*}
    A:=P_0\land\bigwedge_{1\leq r\leq k-1}\lnot P_r
  \end{equation*}
  and for $1\leq j\leq k-1$,
  \begin{equation*}
    C_j:=P_j\land\bigwedge_{\substack{0\leq r\leq k-1\\r\neq j}}\lnot P_r.  
  \end{equation*}
So, the formula $A$ says that only $P_0$ holds among the proposition letters $P_0,P_1,\ldots,P_{k-1}$, and $C_j$ says that only $P_j$ holds. Let
\begin{equation*}
  B_0:=A\land\Diamond_a\top\land\bigwedge_{j=1}^{k-1}\Diamond_b C_j\land\Box_a(A\land\bigwedge_{j=1}^{k-1}\Diamond_b C_j)
\end{equation*}
and for $1\leq j\leq k-1$,
\begin{equation*}
  B_j:=(A\lor C_j)\land\Box_a\bot\land\bigwedge_{m=1}^{j-1}\Box_b\lnot C_m\land\Diamond_b C_j.
\end{equation*}
Finally let $\varphi_k:=\lnot (B_0\lor B_1\lor\cdots\lor B_{k-1})$.

We show that $\varphi_k$ is non-trivially $0^k$-valid but not $0^{k+1}$-valid.

\smallskip\noindent
\textbf{$\varphi_k$ is $0^k$-valid}
Let $M$ be any KD45 model and suppose $M,w\models\lnot\varphi_k$. Then $M,w\models B_0\lor B_1\lor\cdots\lor B_{k-1}$. Our goal is to show $M|^n\varphi_k,w\models\lnot\varphi_k$ for all $1\leq n\leq k-1$. That is,
\begin{equation*}
  M|^n\varphi_k,w\models B_0\lor B_1\lor\cdots\lor B_{k-1}
\end{equation*}
for all $1\leq n\leq k-1$.

We first show $M|^n\varphi_k,w\models\Box_a\bot$ for all $n\geq 1$. Since $M$ is serial, $\Box_a\bot$ is false at $M,w$, so we have $M,w\models B_0$. So, take any $v\in W$ with $wR_a v$. Then, we have $M,v\models A\land \bigwedge_{j=1}^{k-1}\Diamond_b C_j$ by $M,w\models\Box_a(A\land\bigwedge_{j=1}^{k-1}\Diamond_b C_j$). Also, $M,v\models\Diamond_a\top\land\Box_a(A\land\bigwedge_{j=1}^{k-1}\Diamond_b C_j)$ since $M$ is transitive and Euclidean. Thus, we get $M,v\models B_0$. This means that all the $a$-successor arrows from $w$ will be deleted after an update so that we have $M|^n\varphi_k,w\models\Box_a\bot$ for all $n\geq 1$.

Next, for each $1\leq j\leq k-1$, choose an $s_j\in W$ such that $wR_b s_j$ and $M,s_j\models C_j$ (this is possible by $M,w\models B_0$). 
\begin{claim*}
  $w(R_b|^n\varphi_k)s_m$ for all $1\leq n\leq k-1$ and $n\leq m\leq k-1$.
\end{claim*}
\begin{claimproof}
  We show the claim by induction on $n$. 
  
  For $n=1$, fix $1\leq m\leq k-1$. We already have $wR_b s_m$. Also, $M,s_m\models\varphi_k$. In fact, $M,s_m\models\lnot B_0$ because $A$ is false at $M,s_m$ (recall that $M,s_m\models C_m$ and that $A$ says ``only $P_0$ is true'' while $C_m$ says ``only $P_m$ is true''), and $M,s_m\models\lnot B_m$ because $\Box_a\bot$ is false at $M,s_m$. Thus, $w(R_b|\varphi_k)s_m$.

  Assume $w(R_b|^n\varphi_k)s_m$ for some $1\leq n<k-1$ and for all $n\leq m\leq k-1$. Fix $n+1\leq m\leq k-1$. To show $w(R_b|^{n+1}\varphi_k)s_m$, it is enough to show $M|^n\varphi_k,s_m\models\varphi_k$.

  First, we have $M|^n\varphi_k,s_m\models \lnot B_0$ since $M|^n\varphi_k,s_m\models C_m$ and hence $M|^n\varphi_k,s_m\models\lnot A$.

  Second, $M|^n\varphi_k,s_m\models\lnot B_j$ for all $j\neq m$. In fact, $B_j$ requires $A\lor C_j$ but $A$ is impossible by the previous argument while $M|^n\varphi_k,s_m\models C_m$ holds by construction, which is incompatible with $C_j$.

  Finally, $M|^n\varphi_k,s_m\models \lnot B_m$. In fact, since $w(R_b|^n\varphi_k)s_m$ and $w(R_b|^n\varphi_k)s_n$ by the inductive hypothesis, we get $s_m(R_b|^n\varphi_k)s_n$ by Euclideanness. Thus, together with $M,s_n\models C_n$, we have $M|^n\varphi_k,s_m\models\Diamond_b C_n$, namely $M|^n\varphi_k,s_m\not\models\Box_b\lnot C_n$. Thus, $M|^n\varphi_k,s_m\models\lnot B_m$ (note that $1\leq n\leq m-1$).
  
  Therefore, $M|^n\varphi_k,s_m\models\varphi_k$.
\end{claimproof}
Finally, we show $M|^n\varphi_k,w\models\lnot\varphi_k$ for all $1\leq n\leq k-1$. Fix $1\leq n\leq k-1$. By the claim above, we have $w(R_b|^n\varphi_k)s_n$ so there is an $s_j$ such that $w(R_b|^n\varphi_k)s_j$ with $1\leq j\leq n$. Let $m$ be the least such index.

Now, we show $M|^n\varphi_k,w\models B_m$ in particular. First, $M|^n\varphi_k,w\models A$ because $M,w\models B_0$ and valuations never change. Also, $M|^n\varphi_k,w\models\Box_a\bot$ as already shown. Furthermore, $M|^n\varphi_k,w\models\bigwedge_{j=1}^{m-1}\Box_b\lnot C_j\land \Diamond_b C_m$ by minimality of $m$. Thus, $M|^n\varphi_k,w\models B_m$.

Therefore, $M|^n\varphi_k,w\models\lnot\varphi_k$ for all $1\leq n\leq k-1$. Together with the initial assumption that $M,w\models\lnot\varphi_k$, we conclude that $\varphi_k$ is $0^k$-valid.

\smallskip\noindent
\textbf{$\varphi_k$ is not $0^{k+1}$-valid.}
Define the KD45 model $M=(W, \{R_i\}_{i\in G},V)$ by $W=\{w,s_1,\ldots,s_{k-1}\}$, $R_a=W\times \{w\}$, $R_b=W\times \{s_1,\ldots,s_{k-1}\}$, $R_c=W\times W$ for all $c\in G\backslash\{a,b\}$, and finally $V(P_0)=\{w\}$ and $V(P_j)=\{s_j\}$ for $1\leq j\leq k-1$. We then have $M,w\models A$ and $M,s_j\models C_j$ for $1\leq j\leq k-1$.

We check the truth of $\varphi_k$ at each state in the updated models (see Figure~\ref{fig:kd45_countermodel_0k_but_not_0_k+1}).

For the initial model $M$, we have $M,w\models\lnot\varphi_k$ by $M,w\models B_0$. Also, $M,s_j\models\varphi_k$ for all $1\leq j\leq k-1$ since $B_0$ is false by $M,s_j\not\models A$ and $B_j$ is false by $M,s_j\not\models\bigwedge_{m=1}^{j-1}\Box_b\lnot C_m$.

For $M|\varphi_k$, we have $M|\varphi_k,w\models\lnot\varphi_k$ and $M|\varphi_k,s_1\models\lnot\varphi_k$ because $B_1$ holds at $M|\varphi_k,w$ and $M|\varphi_k,s_1$. Also, $M|\varphi_k,s_j\models\varphi_k$ for all $2\leq j\leq k-1$.

For $M|^2\varphi_k$, we have $M|^2\varphi_k,w\models\lnot\varphi_k$ and $M|^2\varphi_k,s_2\models\lnot\varphi_k$ because $B_2$ holds at $M|^2\varphi_k,w$ and $M|^2\varphi_k,s_2$. Also, $M|^2\varphi_k,s_j\models\varphi_k$ for all $3\leq j\leq k-1$.

Continuing this reasoning, we have $M|^n\varphi_k,w\models\lnot\varphi_k$ for all $0\leq n\leq k-1$.

On the other hand, we have $M|^k\varphi_k,w\models\varphi_k$ because $\Diamond_b C_j$ is false for all $1\leq j\leq k-1$. Thus, we have the truth pattern $0^k1^\omega$, so $\varphi_k$ is not $0^{k+1}$-valid.

\end{proof}

\begin{figure}[htbp]
  \centering
  \hspace*{-0.7cm}%
  \begin{tikzpicture}[
      >=Stealth,
      world/.style={circle, fill=black, inner sep=1.45pt},
      acc/.style={->, thick, shorten >=2pt, shorten <=2pt},
      bacc/.style={<->, thick, shorten >=1pt, shorten <=1pt},
      lab/.style={font=\small},
      every loop/.style={looseness=7, min distance=8mm},
      scale=0.65,
      transform shape
    ]


    \begin{scope}[shift={(0,0)}]
      \node[world] (w0) at (2.4,2.7) {};
      \node[world] (s10) at (0,0) {};
      \node[world] (s20) at (1.6,0) {};
      \node (dots0) at (3.2,0) {$\cdots$};
      \node[world] (sk0) at (4.8,0) {};

      \node[lab] at (1.8,3) {$w$};
      \node[lab] at (2.4,3.7) {$\neg\varphi_k$};

      \node[lab] at (0,-1.3) {$s_1$};
      \node[lab] at (1.6,-1.3) {$s_2$};
      \node[lab] at (4.8,-1.3) {$s_{k-1}$};

      \node[lab] at (0,-1.6) {$\varphi_k$};
      \node[lab] at (1.6,-1.6) {$\varphi_k$};
      \node[lab] at (4.8,-1.6) {$\varphi_k$};

      \path[acc]
        (w0) edge[loop above] node[lab, right] {$a$} ()
        (s10) edge[loop below] node[lab, below] {$b$} ()
        (s20) edge[loop below] node[lab, below] {$b$} ()
        (sk0) edge[loop below] node[lab, below] {$b$} ();

      \draw[acc] (s10) to[bend left=24] node[lab, left] {$a$} (w0);
      \draw[acc] (w0) to[bend left=24] node[lab, right] {$b$} (s10);

      \draw[acc] (s20) to[bend left=19] node[lab, left] {$a$} (w0);
      \draw[acc] (w0) to[bend left=19] node[lab, right] {$b$} (s20);

      \draw[acc] (sk0) to[bend left=24] node[lab, right] {$a$} (w0);
      \draw[acc] (w0) to[bend left=24] node[lab, right] {$b$} (sk0);

      \draw[bacc] ($(s10)+(0.18,0)$) -- ($(s20)+(-0.18,0)$);
      \draw[bacc] ($(s20)+(0.18,0)$) -- ($(dots0)+(-0.35,0)$);
      \draw[bacc] ($(dots0)+(0.35,0)$) -- ($(sk0)+(-0.18,0)$);

      \node[font=\large] at (2.4,-1.9) {$M$};
    \end{scope}

    \node[font=\Large] at (7.20,1.10) {$\overset{\varphi_k}{\Longrightarrow}$};

    \begin{scope}[shift={(9.2,0)}]
      \node[world] (w1) at (2.4,2.7) {};
      \node[world] (s11) at (0,0) {};
      \node[world] (s21) at (1.6,0) {};
      \node (dots1) at (3.2,0) {$\cdots$};
      \node[world] (sk1) at (4.8,0) {};

      \node[lab] at (2.4,3.03) {$w$};
      \node[lab] at (2.4,3.40) {$\neg\varphi_k$};

      \node[lab] at (0,-1.3) {$s_1$};
      \node[lab] at (1.6,-1.3) {$s_2$};
      \node[lab] at (4.8,-1.3) {$s_{k-1}$};

      \node[lab] at (0,-1.6) {$\neg\varphi_k$};
      \node[lab] at (1.6,-1.6) {$\varphi_k$};
      \node[lab] at (4.8,-1.6) {$\varphi_k$};

      \path[acc]
        (s11) edge[loop below] node[lab, below] {$b$} ()
        (s21) edge[loop below] node[lab, below] {$b$} ()
        (sk1) edge[loop below] node[lab, below] {$b$} ();

      \draw[acc] (w1) -- node[lab, left] {$b$} (s11);
      \draw[acc] (w1) -- node[lab, left] {$b$} (s21);
      \draw[acc] (w1) -- node[lab, right] {$b$} (sk1);

      \draw[bacc] ($(s11)+(0.18,0)$) -- ($(s21)+(-0.18,0)$);
      \draw[bacc] ($(s21)+(0.18,0)$) -- ($(dots1)+(-0.35,0)$);
      \draw[bacc] ($(dots1)+(0.35,0)$) -- ($(sk1)+(-0.18,0)$);

      \node[font=\large] at (2.4,-1.9) {$M\mid\varphi_k$};
    \end{scope}

    \node[font=\Large] at (16.05,1.10) {$\overset{\varphi_k}{\Longrightarrow}$};

    \begin{scope}[shift={(18.0,0)}]
      \node[world] (w2) at (2.65,2.7) {};
      \node[world] (s12) at (0,0) {};
      \node[world] (s22) at (1.6,0) {};
      \node[world] (s32) at (3.0,0) {};
      \node (dots2) at (4.5,0) {$\cdots$};
      \node[world] (sk2) at (6.0,0) {};

      \node[lab] at (2.65,3.03) {$w$};
      \node[lab] at (2.65,3.40) {$\neg\varphi_k$};

      \node[lab] at (0,-1.3) {$s_1$};
      \node[lab] at (1.6,-1.3) {$s_2$};
      \node[lab] at (3.0,-1.3) {$s_3$};
      \node[lab] at (6.0,-1.3) {$s_{k-1}$};

      \node[lab] at (1.6,-1.6) {$\neg\varphi_k$};
      \node[lab] at (3.0,-1.6) {$\varphi_k$};
      \node[lab] at (6.0,-1.6) {$\varphi_k$};

      \path[acc]
        (s22) edge[loop below] node[lab, below] {$b$} ()
        (s32) edge[loop below] node[lab, below] {$b$} ()
        (sk2) edge[loop below] node[lab, below] {$b$} ();

      \draw[acc] (w2) -- node[lab, left] {$b$} (s22);
      \draw[acc] (w2) -- node[lab, right] {$b$} (s32);
      \draw[acc] (w2) -- node[lab, right] {$b$} (sk2);

      \draw[acc] ($(s12)+(0.18,0)$) -- ($(s22)+(-0.18,0)$);
      \draw[bacc] ($(s22)+(0.18,0)$) -- ($(s32)+(-0.18,0)$);
      \draw[bacc] ($(s32)+(0.18,0)$) -- ($(dots2)+(-0.35,0)$);
      \draw[bacc] ($(dots2)+(0.35,0)$) -- ($(sk2)+(-0.18,0)$);

      \node[font=\large] at (2.65,-1.9) {$M\mid^{2}\varphi_k$};
    \end{scope}


    \node[font=\Large] at (0.55,-5.35) {$\overset{\varphi_k}{\Longrightarrow}$};
    \node[font=\large] at (2.05,-5.35) {$\cdots$};
    \node[font=\Large] at (3.55,-5.35) {$\overset{\varphi_k}{\Longrightarrow}$};

    \begin{scope}[shift={(5.1,-6.2)}]
      \node[world] (wk1) at (2.4,2.7) {};
      \node[world] (t1) at (0,0) {};
      \node[world] (t2) at (1,0) {};
      \node (td1) at (2.4,0) {$\cdots$};
      \node[world] (tkm2) at (3.4,0) {};
      \node[world] (tkm1) at (4.8,0) {};

      \node[lab] at (2.4,3.0) {$w$};
      \node[lab] at (2.4,3.3) {$\neg\varphi_k$};

      \node[lab] at (0,-1.3) {$s_1$};
      \node[lab] at (1,-1.3) {$s_2$};
      \node[lab] at (3.4,-1.3) {$s_{k-2}$};
      \node[lab] at (4.8,-1.3) {$s_{k-1}$};

      \node[lab] at (4.8,-1.6) {$\neg\varphi_k$};

      \path[acc]
        (tkm1) edge[loop below] node[lab, below] {$b$} ();

      \draw[acc] (wk1) -- node[lab, right] {$b$} (tkm1);

      \draw[acc] ($(tkm2)+(0.35,0)$) -- ($(tkm1)+(-0.18,0)$);

      \node[font=\large] at (2.4,-1.9) {$M\mid^{k-1}\varphi_k$};
    \end{scope}

    \node[font=\Large] at (12.35,-5.35) {$\overset{\varphi_k}{\Longrightarrow}$};

    \begin{scope}[shift={(14.25,-6.2)}]
      \node[world] (wk) at (2.4,2.7) {};
      \node[world] (u1) at (0,0) {};
      \node (ud) at (2.4,0) {$\cdots$};
      \node[world] (uk) at (4.8,0) {};

      \node[lab] at (2.4,3.0) {$w$};
      \node[lab] at (2.4,3.3) {$\varphi_k$};

      \node[lab] at (0,-0.38) {$s_1$};
      \node[lab] at (4.8,-0.38) {$s_{k-1}$};

      \node[lab] at (4.8,-0.78) {$\varphi_k$};

      \node[font=\large] at (2.4,-1.45) {$M\mid^{k}\varphi_k$};
    \end{scope}

    \node[font=\Large] at (21.35,-5.35) {$\overset{\varphi_k}{\Longrightarrow}$};
    \node[font=\large] at (22.85,-5.35) {$\cdots$};

  \end{tikzpicture}
  \caption{$\varphi_k$ is not $0^{k+1}$-valid in multi-agent KD45 (Lemma~\ref{lem:0k_but_not_0kplus1_multi_kd45}). Transitive $b$-arrows at the bottom are omitted.}
  \label{fig:kd45_countermodel_0k_but_not_0_k+1}
\end{figure}
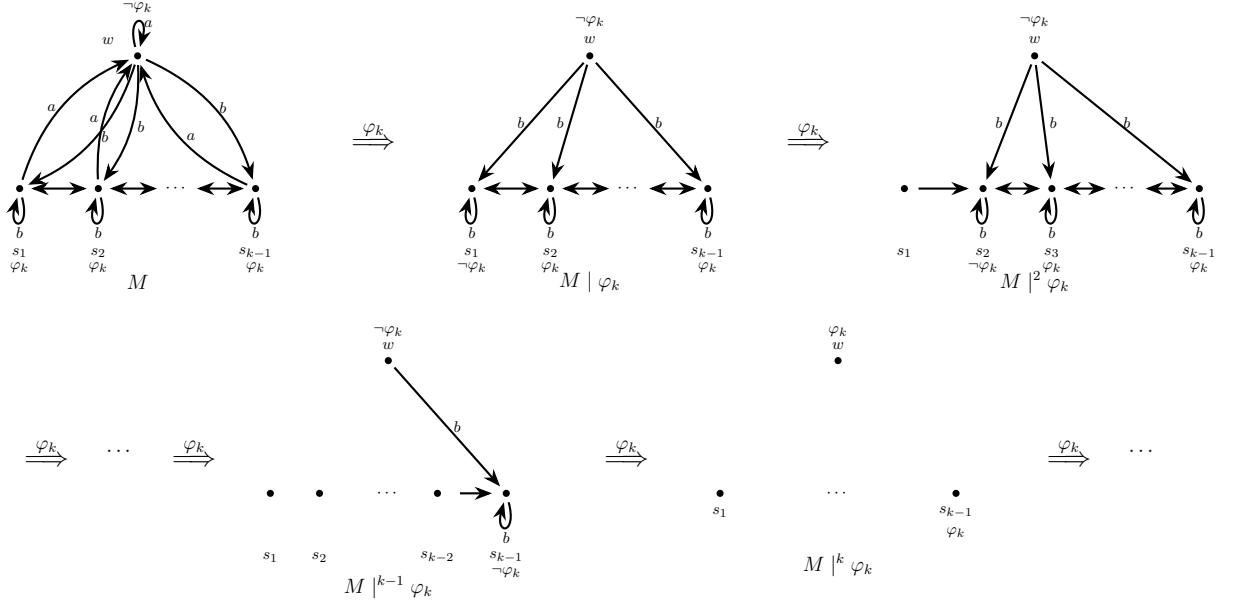

In the following two lemmas, we use the notion of the \emph{type} of a state.
\begin{lemma}[Existence of non-trivially $01^k$-valid but not $01^{k+1}$-valid formula]\label{lem:01k-valid_but_not_01kplus1-valid}
  In multi-agent K45, KD45, and S5: For every $k\geq 1$, there is a formula that is non-trivially $01^k$-valid but not $01^{k+1}$-valid.
\end{lemma}

\begin{proof}

  Fix $k\ge 1$. Let $A_k:=\{r,b,a_0,a_1,\dots,a_{k+1}\}$. For each $t\in A_k$, take a new propositional letter $P_t$, and define the formula
  \begin{equation*}
    \chi_t:=P_t\wedge\bigwedge_{s\in A_k\setminus\{t\}}\neg P_s.
  \end{equation*}
  Then, define the $A_k$-\emph{type}  of state $w$ in a model $M$, denoted $\text{type}_{A_k}\,(w)$ or simply $\text{type}\,(w)$, to be the unique $t\in A_k$ such that $M,w\models\chi_t$ (if such a unique $t$ exists). So, $\chi_t$ says that the $A_k$-type of the current state is $t$.
   For $X\subseteq A_k$, define
  \begin{equation*}
    E_X:=\Box\bigvee_{t\in X}\chi_t\ \wedge\ \bigwedge_{t\in X}\Diamond\chi_t.
  \end{equation*}
  $E_X$ says that the current successor set contains exactly the types in $X$. For $0\le i\le k+1$, let $X_i:=\{b,a_i,a_{i+1},\dots,a_{k+1}\}$, $X_{k+2}:=\{b\}$, and $Y_0:=X_0\cup\{r\}$. Now define
  \begin{equation*}
    B_k:= \bigl(\chi_r\wedge(E_{Y_0}\vee E_{X_{k+1}})\bigr) \vee (\chi_{a_0}\land E_{Y_0}) \lor\bigvee_{i=1}^{k+1}\bigl(\chi_{a_i}\wedge E_{X_i}\bigr),
  \end{equation*} and finally $\varphi_k:=\neg B_k$.
  $B_k$ says that either (a) the current state has type $r$ and the successor types are exactly in $Y_0$, (b) the current state has type $r$ and the successor types are exactly in $X_{k+1}$, (c) the current state has type $a_0$ and the successor types are exactly in $Y_0$, or (d) the current state has type $a_i$ and the successor types are exactly in $X_i$ for some $1\leq i \leq k+1$.

\smallskip\noindent\textbf{$\varphi_k$ is $01^k$-valid}
  We prove that $\varphi_k$ is $01^k$-valid. Let $M=(W,\{R_i\}_{i\in G},V)$ be any K45 model and suppose $M,w\models\neg\varphi_k$, i.e., $M,w\models B_k$.  Note that if $wR_iv$, then $R_i(w)=R_i(v)$ for all $i\in G$ and $w,v\in W$ by Lemma~\ref{lem:modal agreement lemma}, so the successor types of $w$ agree with those of $v$.

  \smallskip\noindent(a) If $\text{type}\,(w)=r$ and $M,w\models E_{Y_0}$, $\varphi_k$ is false only at $r$ and $a_0$ among the types in $Y_0$ (see the model $M$ in Figure~\ref{fig:type_transition_for_a__01k_but_not_01kplus1}). Thus, after the announcement of $\varphi_k$, only the types $r$ and $a_0$ are deleted, changing the successor type $Y_0$ to $X_1$. Also, $M|\varphi_k,w\models\varphi_k$ since $\text{type}\,(w)=r$ but $w$ satisfies neither $E_{Y_0}$ nor $E_{X_{k+1}}$. Repeating this, the successor type eventually becomes $X_{k+2}=\{b\}$ and stabilizes there. In this process, we have the truth pattern $01^k01^\omega$ at $w$.

  \smallskip\noindent
  (b) If $\text{type}\,(w)=r$ and $M,w\models E_{X_{k+1}}$, we have the truth pattern $01^\omega$ as in Figure~\ref{fig:type_transition_for_b_01k_but_not_01kplus1}.

  \smallskip\noindent
  (c) If $\text{type}\,(w)=a_0$ and $M,w\models E_{Y_0}$, we have the truth pattern $01^\omega$ as in Figure~\ref{fig:type_transition_for_c_01k_but_not_01kplus1} (with some modification to the figure).

  \smallskip\noindent
  (d) If $\text{type}\,(w)=a_i$ and $M,w\models E_{X_i}$ for some $0\leq i\leq k+1$, we have the truth pattern $01^\omega$ as in Figure~\ref{fig:type_transition_for_c_01k_but_not_01kplus1}.

  Therefore, $\varphi_k$ is indeed $01^k$-valid.

  \smallskip\noindent\textbf{$\varphi_k$ is not $01^{k+1}$-valid}
  Next, we show that $\varphi_k$ is not $01^{k+1}$-valid. Define the S5 model $M=(W,R, V)$ by $W=\{w_r, w_{b},w_{a_0},\ldots,w_{a_{k+1}}\}$, $R=W\times W$, and $V(P_t)=\{w_t\}$ for each $t\in A_k$. Then, the truth of $\varphi_k$ at $w_r$ follows $01^k01^\omega$ as in Figure~\ref{fig:type_transition_for_a__01k_but_not_01kplus1}, so $\varphi_k$ is not $01^{k+1}$-valid.

  Thus, there is a formula that is non-trivially $01^k$-valid but not $01^{k+1}$-valid.
\end{proof}

\begin{figure}[htbp]
  \centering
  \hspace*{-0.7cm}%
  \begin{subfigure}[t]{\linewidth}
  \begin{tikzpicture}[
      >=Stealth,
      world/.style={circle, fill=black, inner sep=1.25pt},
      edge/.style={->, thick, shorten >=2pt, shorten <=2pt},
      capsule/.style={draw, rounded corners=7pt, inner xsep=13pt, inner ysep=6pt},
      tlabel/.style={inner sep=1pt},
      vlabel/.style={inner sep=1pt},
      xlab/.style={font=\small},
      scale=0.65,
      transform shape
    ]


    \begin{scope}[shift={(0,0)}]
      \CaseAState{r0}{(-0.95,0.25)}{r}{\neg\varphi_k}
      \CaseAState{b0}{(0,0.25)}{b}{\varphi_k}
      \CaseAState{a00}{(0.95,0.25)}{a_0}{\neg\varphi_k}
      \CaseAState{a10}{(1.90,0.25)}{a_1}{\varphi_k}
      \node (dots0) at (2.65,0.25) {$\cdots$};
      \CaseAState{akp10}{(3.60,0.25)}{a_{k+1}}{\varphi_k}
      \CaseARoot{w0}{(1.78,-1.45)}{r}{\neg\varphi_k}
      \node[capsule, fit=(r0)(b0)(a00)(a10)(dots0)(akp10)] (cap0) {};
      \node[xlab, left=2pt of cap0] {$Y_0$};
      \draw[edge] (w0) -- (r0);
      \draw[edge] (w0) -- (b0);
      \draw[edge] (w0) -- (a00);
      \draw[edge] (w0) -- (a10);
      \draw[edge] (w0) -- (akp10);
      \node[font=\large] at (1.78,-2.55) {$M$};
    \end{scope}

    \node[font=\Large] at (5.45,-0.15) {$\overset{\varphi_k}{\Longrightarrow}$};

    \begin{scope}[shift={(6.55,0)}]
      \CaseAState{b1}{(0,0.25)}{b}{\varphi_k}
      \CaseAState{a11}{(0.95,0.25)}{a_1}{\neg\varphi_k}
      \CaseAState{a21}{(1.90,0.25)}{a_2}{\varphi_k}
      \node (dots1) at (2.65,0.25) {$\cdots$};
      \CaseAState{akp11}{(3.60,0.25)}{a_{k+1}}{\varphi_k}
      \CaseARoot{w1}{(1.78,-1.45)}{r}{\varphi_k}
      \node[capsule, fit=(b1)(a11)(a21)(dots1)(akp11)] (cap1) {};
      \node[xlab, left=2pt of cap1] {$X_1$};
      \draw[edge] (w1) -- (b1);
      \draw[edge] (w1) -- (a11);
      \draw[edge] (w1) -- (a21);
      \draw[edge] (w1) -- (akp11);
      \node[font=\large] at (1.78,-2.55) {$M\mid\varphi_k$};
    \end{scope}

    \node[font=\Large] at (11.95,-0.15) {$\overset{\varphi_k}{\Longrightarrow}$};
    \node[font=\large] at (13.10,-0.15) {$\cdots$};
    \node[font=\Large] at (14.20,-0.15) {$\overset{\varphi_k}{\Longrightarrow}$};

    \begin{scope}[shift={(15.35,0)}]
      \CaseAState{bk}{(0,0.25)}{b}{\varphi_k}
      \CaseAState{ak}{(1.10,0.25)}{a_k}{\neg\varphi_k}
      \CaseAState{akp1k}{(2.20,0.25)}{a_{k+1}}{\varphi_k}
      \CaseARoot{wk}{(1.10,-1.45)}{r}{\varphi_k}
      \node[capsule, fit=(bk)(ak)(akp1k)] (capk) {};
      \node[xlab, left=2pt of capk] {$X_k$};
      \draw[edge] (wk) -- (bk);
      \draw[edge] (wk) -- (ak);
      \draw[edge] (wk) -- (akp1k);
      \node[font=\large] at (1.10,-2.55) {$M\mid^{k}\varphi_k$};
    \end{scope}


    \node[font=\Large] at (-0.35,-5.35) {$\overset{\varphi_k}{\Longrightarrow}$};

    \begin{scope}[shift={(2.05,-5.20)}]
      \CaseAState{bkp1}{(0,0)}{b}{\varphi_k}
      \CaseAState{akp1}{(1.55,0)}{a_{k+1}}{\neg\varphi_k}
      \CaseARoot{wkp1}{(0.78,-1.45)}{r}{\neg\varphi_k}
      \node[capsule, fit=(bkp1)(akp1)] (capkp1) {};
      \node[xlab, left=2pt of capkp1] {$X_{k+1}$};
      \draw[edge] (wkp1) -- (bkp1);
      \draw[edge] (wkp1) -- (akp1);
      \node[font=\large] at (0.78,-2.55) {$M\mid^{k+1}\varphi_k$};
    \end{scope}

    \node[font=\Large] at (5.20,-5.35) {$\overset{\varphi_k}{\Longrightarrow}$};

    \begin{scope}[shift={(8.00,-5.20)}]
      \CaseAState{bkp2}{(0,0)}{b}{\varphi_k}
      \CaseARoot{wkp2}{(0,-1.45)}{r}{\varphi_k}
      \node[capsule, fit=(bkp2)] (capkp2) {};
      \node[xlab, left=2pt of capkp2] {$X_{k+2}$};
      \draw[edge] (wkp2) -- (bkp2);
      \node[font=\large] at (0,-2.55) {$M\mid^{k+2}\varphi_k$};
    \end{scope}

    \node[font=\Large] at (10.15,-5.35) {$\overset{\varphi_k}{\Longrightarrow}$};

    \begin{scope}[shift={(13.55,-5.20)}]
      \CaseAState{bkp3}{(0,0)}{b}{\varphi_k}
      \CaseARoot{wkp3}{(0,-1.45)}{r}{\varphi_k}
      \node[capsule, fit=(bkp3)] (capkp3) {};
      \node[xlab, left=2pt of capkp3] {$X_{k+2}$};
      \draw[edge] (wkp3) -- (bkp3);
      \node[font=\large] at (0,-2.55) {$M\mid^{k+3}\varphi_k$};
    \end{scope}

    \node[font=\Large] at (15.95,-5.35) {$\overset{\varphi_k}{\Longrightarrow}$};
    \node[font=\large] at (17.25,-5.35) {$\cdots$};

  \end{tikzpicture}
  \caption{The case with $\text{type} (w)=r$ and $M,w\models E_{Y_0}$}
  \label{fig:type_transition_for_a__01k_but_not_01kplus1}
\end{subfigure}
\medskip

\begin{subfigure}[t]{\linewidth}
  \centering

  \begin{tikzpicture}[
      >=Stealth,
      world/.style={circle, fill=black, inner sep=1.25pt},
      edge/.style={->, thick, shorten >=2pt, shorten <=2pt},
      capsule/.style={draw, rounded corners=7pt, inner xsep=13pt, inner ysep=6pt},
      tlabel/.style={inner sep=1pt},
      vlabel/.style={inner sep=1pt},
      xlab/.style={font=\small},
      scale=0.65,
      transform shape
    ]

    \begin{scope}[shift={(0,0)}]
      \CaseAState{bb0}{(0,0)}{b}{\varphi_k}
      \CaseAState{bakp10}{(1.55,0)}{a_{k+1}}{\neg\varphi_k}
      \CaseARoot{bw0}{(0.78,-1.45)}{r}{\neg\varphi_k}
      \node[capsule, fit=(bb0)(bakp10)] (bcap0) {};
      \node[xlab, left=2pt of bcap0] {$X_{k+1}$};
      \draw[edge] (bw0) -- (bb0);
      \draw[edge] (bw0) -- (bakp10);
      \node[font=\large] at (0.78,-2.55) {$M$};
    \end{scope}

    \node[font=\Large] at (3.25,-0.15) {$\overset{\varphi_k}{\Longrightarrow}$};

    \begin{scope}[shift={(6.00,0)}]
      \CaseAState{bb1}{(0,0)}{b}{\varphi_k}
      \CaseARoot{bw1}{(0,-1.45)}{r}{\varphi_k}
      \node[capsule, fit=(bb1)] (bcap1) {};
      \node[xlab, left=2pt of bcap1] {$X_{k+2}$};
      \draw[edge] (bw1) -- (bb1);
      \node[font=\large] at (0,-2.55) {$M\mid\varphi_k$};
    \end{scope}

    \node[font=\Large] at (7.35,-0.15) {$\overset{\varphi_k}{\Longrightarrow}$};

    \begin{scope}[shift={(10.0,0)}]
      \CaseAState{bb2}{(0,0)}{b}{\varphi_k}
      \CaseARoot{bw2}{(0,-1.45)}{r}{\varphi_k}
      \node[capsule, fit=(bb2)] (bcap2) {};
      \node[xlab, left=2pt of bcap2] {$X_{k+2}$};
      \draw[edge] (bw2) -- (bb2);
      \node[font=\large] at (0,-2.55) {$M\mid^{2}\varphi_k$};
    \end{scope}

    \node[font=\Large] at (11.45,-0.15) {$\overset{\varphi_k}{\Longrightarrow}$};
    \node[font=\large] at (12.70,-0.15) {$\cdots$};
  \end{tikzpicture}
  \caption{The case with $\text{type}\,(w)=r$ and $M,w\models E_{X_{k+1}}$}
  \label{fig:type_transition_for_b_01k_but_not_01kplus1}
\end{subfigure}
\medskip
\begin{subfigure}[t]{\linewidth}
  \centering
  \hspace*{-2.0cm}%

  \begin{tikzpicture}[
      >=Stealth,
      world/.style={circle, fill=black, inner sep=1.25pt},
      edge/.style={->, thick, shorten >=2pt, shorten <=2pt},
      capsule/.style={draw, rounded corners=7pt, inner xsep=13pt, inner ysep=6pt},
      tlabel/.style={inner sep=1pt},
      vlabel/.style={inner sep=1pt},
      xlab/.style={font=\small},
      scale=0.55,
      transform shape
    ]

    \begin{scope}[shift={(0,0)}]
      \CaseAState{cb0}{(0,0.3)}{b}{\varphi_k}
      \CaseAState{cai0}{(0.95,0.3)}{a_i}{\neg\varphi_k}
      \CaseAState{caip10}{(2.00,0.3)}{a_{i+1}}{\varphi_k}
      \node (cdots0) at (2.75,0.22) {$\cdots$};
      \CaseAState{cakp10}{(3.70,0.3)}{a_{k+1}}{\varphi_k}
      \CaseARoot{cw0}{(1.90,-1.45)}{a_i}{\neg\varphi_k}
      \node[capsule, fit=(cb0)(cai0)(caip10)(cdots0)(cakp10)] (ccap0) {};
      \node[xlab, left=2pt of ccap0] {$X_i$};
      \draw[edge] (cw0) -- (cb0);
      \draw[edge] (cw0) -- (cai0);
      \draw[edge] (cw0) -- (caip10);
      \draw[edge] (cw0) -- (cakp10);
      \node[font=\large] at (1.90,-2.55) {$M$};
    \end{scope}

    \node[font=\Large] at (5.55,-0.15) {$\overset{\varphi_k}{\Longrightarrow}$};

    \begin{scope}[shift={(6.85,0)}]
      \CaseAState{cb1}{(0,0.3)}{b}{\varphi_k}
      \CaseAState{caip11}{(1.05,0.3)}{a_{i+1}}{\neg\varphi_k}
      \CaseAState{caip21}{(2.10,0.3)}{a_{i+2}}{\varphi_k}
      \node (cdots1) at (2.85,0.22) {$\cdots$};
      \CaseAState{cakp11}{(3.80,0.3)}{a_{k+1}}{\varphi_k}
      \CaseARoot{cw1}{(1.95,-1.45)}{a_i}{\varphi_k}
      \node[capsule, fit=(cb1)(caip11)(caip21)(cdots1)(cakp11)] (ccap1) {};
      \node[xlab, left=2pt of ccap1, yshift=6pt] {$X_{i+1}$};
      \draw[edge] (cw1) -- (cb1);
      \draw[edge] (cw1) -- (caip11);
      \draw[edge] (cw1) -- (caip21);
      \draw[edge] (cw1) -- (cakp11);
      \node[font=\large] at (1.95,-2.55) {$M\mid\varphi_k$};
    \end{scope}

    \node[font=\Large] at (12.60,-0.15) {$\overset{\varphi_k}{\Longrightarrow}$};
    \node[font=\large] at (13.85,-0.15) {$\cdots$};
    \node[font=\Large] at (15.10,-0.15) {$\overset{\varphi_k}{\Longrightarrow}$};

    \begin{scope}[shift={(16.85,0)}]
      \CaseAState{cbfin}{(0,0)}{b}{\varphi_k}
      \CaseARoot{cwfin}{(0,-1.45)}{a_i}{\varphi_k}
      \node[capsule, fit=(cbfin)] (ccapfin) {};
      \node[xlab, left=2pt of ccapfin,yshift=7pt] {$X_{k+2}$};
      \draw[edge] (cwfin) -- (cbfin);
      \node[font=\large] at (0,-2.55) {$M\mid^{k+2-i}\varphi_k$};
    \end{scope}

    \node[font=\Large] at (19.20,-0.15) {$\overset{\varphi_k}{\Longrightarrow}$};
    \node[font=\large] at (20.45,-0.15) {$\cdots$};

  \end{tikzpicture}
  \caption{The case with $\text{type}\,(w)=a_0$ and $M,w\models E_{Y_0}$, or (d) the case with $\text{type}\,(w)=a_i$ and $M,w\models E_{X_i}$ for some $1\leq i\leq k+1$. Here when $i=0$, the $X_i$ in the left model is replaced with $Y_0$ and there is a dot with labels $r$ and $\lnot\varphi_k$ inside the box.}
  \label{fig:type_transition_for_c_01k_but_not_01kplus1}
\end{subfigure}
\caption{$\varphi_k$ is non-trivially $01^k$-valid but not $01^{k+1}$-valid in multi-agent K45, KD45, and S5 (Lemma~\ref{lem:01k-valid_but_not_01kplus1-valid}).}
\label{fig: 0_k_but_not_0_k+1_kd45}

\end{figure}
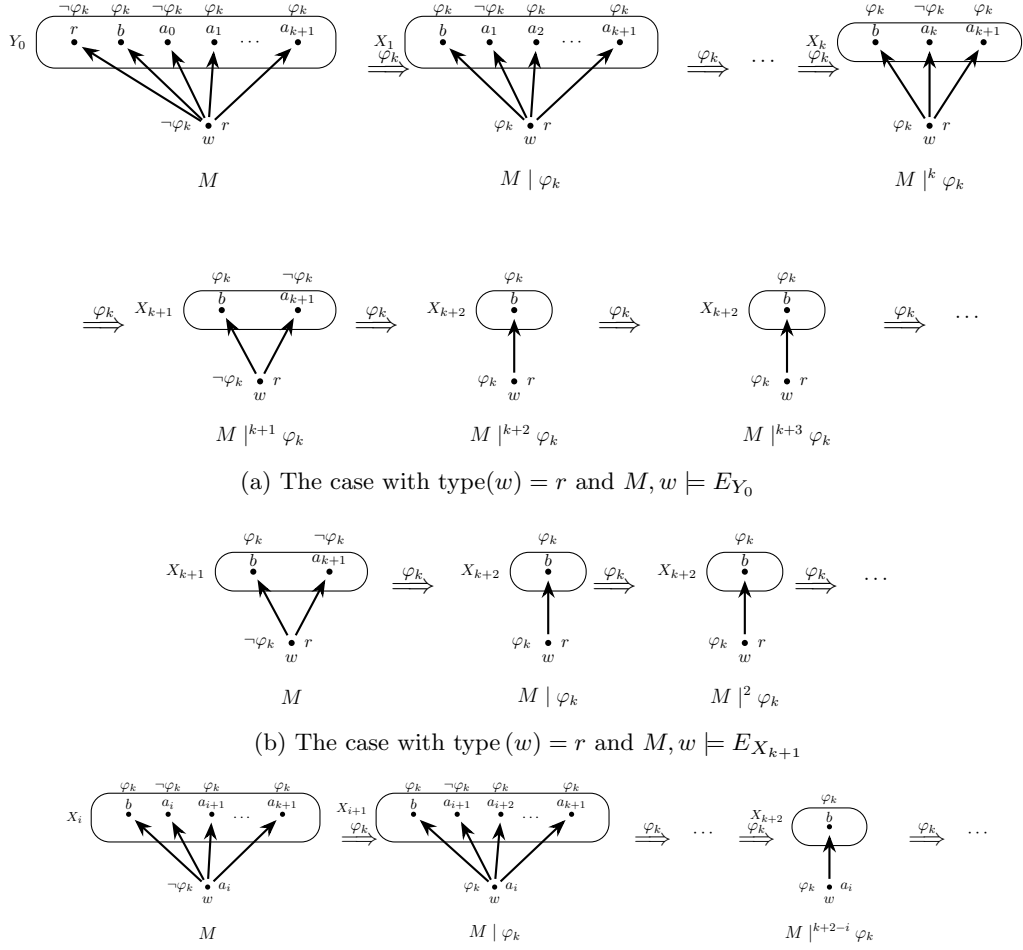

\begin{lemma}[Existence of non-trivially $0^k$-valid but not $0^{k+1}$-valid formula in S5]\label{lem:0k-valid_but_not_0_plus1-validity}
  In multi-agent S5: For every $k\geq 2$, there is a formula that is non-trivially $0^k$-valid but not $0^{k+1}$-valid.
\end{lemma}
\begin{proof}
  Fix $k \geq 2$ and let $A_k := \{a_0,a_1,\ldots,a_k\}$. For each $t\in A_k$, take a new propositional letter $P_t$, and define $\chi_t := P_t \wedge \bigwedge_{s\in A_k\setminus\{t\}}\neg P_s$. For each $X\subseteq A_k$, define $E_X := \Box \bigvee_{t\in X}\chi_t \wedge \bigwedge_{t\in X}\Diamond\chi_t$ with the convention that $E_\emptyset := \Box\bot$. For $0\leq i\leq k$, let $X_i := \{a_i,a_{i+1},\ldots,a_k\}$ and $X_{k+1}:=\emptyset$. Define the formula
  \begin{equation*} B_k := \left( \chi_{a_0}\wedge \bigvee_{\ell=0}^{k-1} E_{X_\ell} \right) \vee \left(\chi_{a_1}\land\bigvee_{\ell=1}^k E_{X_\ell}\right)\lor\bigvee_{i=2}^{k} \left( \chi_{a_i}\wedge \bigvee_{\ell=i}^{k+1} E_{X_\ell} \right)
  \end{equation*} and finally let $\varphi_k := \neg B_k$.

\smallskip\noindent\textbf{$\varphi_k$ is $0^k$-valid}
  We first show that $\varphi_k$ is $0^k$-valid. Let $M$ be any S5 model and suppose that $M,w\models\neg\varphi_k$, i.e. $M,w\models B_k$. Note that for all $0\leq i\leq k$, if $\text{type}(w)=a_i$, then we actually have
  \begin{equation*}
    M,w\models\chi_{a_i}\land E_{X_i}
  \end{equation*}
   since the successor type set of $w$ must contain $a_i$ by reflexivity.

  (a) If $\text{type}(w)=a_0$ and $M,w\models E_{X_0}$, $\varphi_k$ is false exactly at $w$ and $a_0$ in $M$ so after the announcement of $\varphi_k$, the successor type set of $w$ becomes $X_1$ as in Figure~\ref{fig:0k-valid_but_not_0_kplus1-valid_a}. Repeating this, we have the truth pattern $0^k1^\omega$-valid.

  (b) If $\text{type}(w)=a_1$ and $M,w\models E_{X_1}$, we have the truth pattern $0^k1^\omega$ as in Figure~\ref{fig:0k-valid_but_not_0_kplus1-valid_b}.

  (c) If $\text{type}(w)=a_i$ and $M,w\models E_{X_\ell}$ for some $2\leq i\leq k$ and $i\leq \ell\leq k+1$, we have the truth pattern $0^\omega$ as in Figure~\ref{fig:0k-valid_but_not_0_kplus1-valid_c}.

  Thus, $\varphi_k$ is $0^k$-valid.

\smallskip\noindent\textbf{$\varphi_k$ is not $0^{k+1}$-valid}
  It remains to show that $\varphi_k$ is $0^k$-satisfiable but not $0^{k+1}$-valid. Take the S5 model $M=(W,R,V)$ such that $W=\{w_{a_0},w_{a_1},\ldots,w_{a_k}\}$, $R=W\times W$, and $V(P_{a_i})=\{w_{a_i}\}$ for all $0\leq i\leq k$. Then, the announcements proceed as in case (a) so that $\varphi_k$ is $0^k$-satisfiable but not $0^{k+1}$-valid.

  Thus, $\varphi_k$ is non-trivially $0^k$-valid but not $0^{k+1}$-valid in multi-agent S5.  
\end{proof}

\begin{figure}[htbp]
  \centering

  \begin{subfigure}[t]{\linewidth}
    \hspace*{0.7cm}%
  \begin{tikzpicture}[
      >=Stealth,
      world/.style={circle, fill=black, inner sep=1.25pt},
      edge/.style={->, thick, shorten >=2pt, shorten <=2pt},
      capsule/.style={draw, rounded corners=7pt, inner xsep=13pt, inner ysep=6pt},
      tlabel/.style={inner sep=1pt},
      vlabel/.style={inner sep=1pt},
      xlab/.style={font=\small},
      scale=0.61,
      transform shape
    ]


    \begin{scope}[shift={(0,0)}]
      \CaseAState{za00}{(0,0.25)}{a_0}{\neg\varphi_k}
      \CaseAState{za10}{(0.95,0.25)}{a_1}{\varphi_k}
      \node (zdots0) at (1.75,0.25) {$\cdots$};
      \CaseAState{zak0}{(2.65,0.25)}{a_k}{\varphi_k}
      \CaseARoot{zw0}{(1.32,-1.45)}{a_0}{\neg\varphi_k}
      \node[capsule, fit=(za00)(za10)(zdots0)(zak0)] (zcap0) {};
      \node[xlab, left=2pt of zcap0] {$X_0$};
      \draw[edge] (zw0) -- (za00);
      \draw[edge] (zw0) -- (za10);
      \draw[edge] (zw0) -- (zak0);
      \node[font=\large] at (1.32,-2.55) {$M$};
    \end{scope}

    \node[font=\Large] at (4.25,-0.15) {$\overset{\varphi_k}{\Longrightarrow}$};

    \begin{scope}[shift={(5.25,0)}]
      \CaseAState{za11}{(0,0.25)}{a_1}{\neg\varphi_k}
      \CaseAState{za21}{(0.95,0.25)}{a_2}{\varphi_k}
      \node (zdots1) at (1.75,0.25) {$\cdots$};
      \CaseAState{zak1}{(2.65,0.25)}{a_k}{\varphi_k}
      \CaseARoot{zw1}{(1.32,-1.45)}{a_0}{\neg\varphi_k}
      \node[capsule, fit=(za11)(za21)(zdots1)(zak1)] (zcap1) {};
      \node[xlab, left=2pt of zcap1] {$X_1$};
      \draw[edge] (zw1) -- (za11);
      \draw[edge] (zw1) -- (za21);
      \draw[edge] (zw1) -- (zak1);
      \node[font=\large] at (1.32,-2.55) {$M\mid\varphi_k$};
    \end{scope}

    \node[font=\Large] at (9.55,-0.15) {$\overset{\varphi_k}{\Longrightarrow}$};
    \node[font=\large] at (10.65,-0.15) {$\cdots$};
    \node[font=\Large] at (11.75,-0.15) {$\overset{\varphi_k}{\Longrightarrow}$};

    \begin{scope}[shift={(12.85,0)}]
      \CaseAState{zakm1}{(0,0.25)}{a_{k-1}}{\neg\varphi_k}
      \CaseAState{zakm0}{(1.55,0.25)}{a_k}{\varphi_k}
      \CaseARoot{zwkm1}{(0.78,-1.45)}{a_0}{\neg\varphi_k}
      \node[capsule, fit=(zakm1)(zakm0)] (zcapkm1) {};
      \node[xlab, left=2pt of zcapkm1] {$X_{k-1}$};
      \draw[edge] (zwkm1) -- (zakm1);
      \draw[edge] (zwkm1) -- (zakm0);
      \node[font=\large] at (0.78,-2.55) {$M\mid^{k-1}\varphi_k$};
    \end{scope}


    \node[font=\Large] at (-0.35,-5.35) {$\overset{\varphi_k}{\Longrightarrow}$};

    \begin{scope}[shift={(2.05,-5.20)}]
      \CaseAState{zak}{(0,0)}{a_k}{\neg\varphi_k}
      \CaseARoot{zwk}{(0,-1.45)}{a_0}{\varphi_k}
      \node[capsule, fit=(zak)] (zcapk) {};
      \node[xlab, left=2pt of zcapk] {$X_k$};
      \draw[edge] (zwk) -- (zak);
      \node[font=\large] at (0,-2.55) {$M\mid^{k}\varphi_k$};
    \end{scope}

    \node[font=\Large] at (5.20,-5.35) {$\overset{\varphi_k}{\Longrightarrow}$};

    \begin{scope}[shift={(8.00,-5.20)}]
      \node[xlab] at (0,0) {$X_{k+1}(=\emptyset)$};
      \CaseARoot{zwkp1}{(0,-1.45)}{a_0}{\varphi_k}
      \node[font=\large] at (0,-2.55) {$M\mid^{k+1}\varphi_k$};
    \end{scope}

    \node[font=\Large] at (10.15,-5.35) {$\overset{\varphi_k}{\Longrightarrow}$};
    \node[font=\large] at (11.45,-5.35) {$\cdots$};

  \end{tikzpicture}
  \caption{The case with $\text{type}(w)=a_0$ and $M,w\models E_{X_0}$}
  \label{fig:0k-valid_but_not_0_kplus1-valid_a}
\end{subfigure}

\medskip

\begin{subfigure}[t]{\linewidth}
  \centering
  \hspace*{-0.7cm}%
  \begin{tikzpicture}[
      >=Stealth,
      world/.style={circle, fill=black, inner sep=1.25pt},
      edge/.style={->, thick, shorten >=2pt, shorten <=2pt},
      capsule/.style={draw, rounded corners=7pt, inner xsep=13pt, inner ysep=6pt},
      tlabel/.style={inner sep=1pt},
      vlabel/.style={inner sep=1pt},
      xlab/.style={font=\small},
      scale=0.61,
      transform shape
    ]


    \begin{scope}[shift={(0,0)}]
      \CaseAState{ba10}{(0,0.25)}{a_1}{\neg\varphi_k}
      \CaseAState{ba20}{(0.95,0.25)}{a_2}{\varphi_k}
      \node (bdots0) at (1.75,0.25) {$\cdots$};
      \CaseAState{bak0}{(2.65,0.25)}{a_k}{\varphi_k}
      \CaseARoot{bw0}{(1.32,-1.45)}{a_1}{\neg\varphi_k}
      \node[capsule, fit=(ba10)(ba20)(bdots0)(bak0)] (bcap0) {};
      \node[xlab, left=2pt of bcap0] {$X_1$};
      \draw[edge] (bw0) -- (ba10);
      \draw[edge] (bw0) -- (ba20);
      \draw[edge] (bw0) -- (bak0);
      \node[font=\large] at (1.32,-2.55) {$M$};
    \end{scope}

    \node[font=\Large] at (4.25,-0.15) {$\overset{\varphi_k}{\Longrightarrow}$};

    \begin{scope}[shift={(5.25,0)}]
      \CaseAState{ba21}{(0,0.25)}{a_2}{\neg\varphi_k}
      \CaseAState{ba31}{(0.95,0.25)}{a_3}{\varphi_k}
      \node (bdots1) at (1.75,0.25) {$\cdots$};
      \CaseAState{bak1}{(2.65,0.25)}{a_k}{\varphi_k}
      \CaseARoot{bw1}{(1.32,-1.45)}{a_1}{\neg\varphi_k}
      \node[capsule, fit=(ba21)(ba31)(bdots1)(bak1)] (bcap1) {};
      \node[xlab, left=2pt of bcap1] {$X_2$};
      \draw[edge] (bw1) -- (ba21);
      \draw[edge] (bw1) -- (ba31);
      \draw[edge] (bw1) -- (bak1);
      \node[font=\large] at (1.32,-2.55) {$M\mid\varphi_k$};
    \end{scope}

    \node[font=\Large] at (9.55,-0.15) {$\overset{\varphi_k}{\Longrightarrow}$};
    \node[font=\large] at (10.65,-0.15) {$\cdots$};
    \node[font=\Large] at (11.75,-0.15) {$\overset{\varphi_k}{\Longrightarrow}$};

    \begin{scope}[shift={(12.85,0)}]
      \CaseAState{bakm1}{(0,0.25)}{a_{k-1}}{\neg\varphi_k}
      \CaseAState{bakm0}{(1.55,0.25)}{a_k}{\varphi_k}
      \CaseARoot{bwkm2}{(0.78,-1.45)}{a_1}{\neg\varphi_k}
      \node[capsule, fit=(bakm1)(bakm0)] (bcapkm1) {};
      \node[xlab, left=2pt of bcapkm1] {$X_{k-1}$};
      \draw[edge] (bwkm2) -- (bakm1);
      \draw[edge] (bwkm2) -- (bakm0);
      \node[font=\large] at (0.78,-2.55) {$M\mid^{k-2}\varphi_k$};
    \end{scope}


    \node[font=\Large] at (-0.35,-5.35) {$\overset{\varphi_k}{\Longrightarrow}$};

    \begin{scope}[shift={(2.05,-5.20)}]
      \CaseAState{bak}{(0,0)}{a_k}{\neg\varphi_k}
      \CaseARoot{bwk}{(0,-1.45)}{a_1}{\neg\varphi_k}
      \node[capsule, fit=(bak)] (bcapk) {};
      \node[xlab, left=2pt of bcapk] {$X_k$};
      \draw[edge] (bwk) -- (bak);
      \node[font=\large] at (0,-2.55) {$M\mid^{k-1}\varphi_k$};
    \end{scope}

    \node[font=\Large] at (5.20,-5.35) {$\overset{\varphi_k}{\Longrightarrow}$};

    \begin{scope}[shift={(8.00,-5.20)}]
      \node[xlab] at (0,0) {$X_{k+1}(=\emptyset)$};
      \CaseARoot{bwkp1}{(0,-1.45)}{a_1}{\varphi_k}
      \node[font=\large] at (0,-2.55) {$M\mid^{k}\varphi_k$};
    \end{scope}

    \node[font=\Large] at (10.15,-5.35) {$\overset{\varphi_k}{\Longrightarrow}$};
    \node[font=\large] at (11.45,-5.35) {$\cdots$};

  \end{tikzpicture}
  \caption{The case with $\text{type}(w)=a_1$ and $M,w\models E_{X_1}$}
  \label{fig:0k-valid_but_not_0_kplus1-valid_b}
\end{subfigure}
\medskip
\begin{subfigure}[t]{\linewidth}
  \centering
  \hspace*{-0.7cm}%
  \begin{tikzpicture}[
      >=Stealth,
      world/.style={circle, fill=black, inner sep=1.25pt},
      edge/.style={->, thick, shorten >=2pt, shorten <=2pt},
      capsule/.style={draw, rounded corners=7pt, inner xsep=13pt, inner ysep=6pt},
      tlabel/.style={inner sep=1pt},
      vlabel/.style={inner sep=1pt},
      xlab/.style={font=\small},
      scale=0.61,
      transform shape
    ]


    \begin{scope}[shift={(0,0)}]
      \CaseAState{cai0}{(0,0.25)}{a_i}{\neg\varphi_k}
      \CaseAState{caip10}{(1.15,0.25)}{a_{i+1}}{\varphi_k}
      \node (cdots0) at (2.05,0.25) {$\cdots$};
      \CaseAState{cak0}{(2.95,0.25)}{a_k}{\varphi_k}
      \CaseARoot{cw0}{(1.48,-1.45)}{a_i}{\neg\varphi_k}
      \node[capsule, fit=(cai0)(caip10)(cdots0)(cak0)] (ccap0) {};
      \node[xlab, left=2pt of ccap0] {$X_i$};
      \draw[edge] (cw0) -- (cai0);
      \draw[edge] (cw0) -- (caip10);
      \draw[edge] (cw0) -- (cak0);
      \node[font=\large] at (1.48,-2.55) {$M$};
    \end{scope}

    \node[font=\Large] at (4.55,-0.15) {$\overset{\varphi_k}{\Longrightarrow}$};

    \begin{scope}[shift={(5.65,0)}]
      \CaseAState{caip11}{(0,0.25)}{a_{i+1}}{\neg\varphi_k}
      \CaseAState{caip21}{(1.30,0.25)}{a_{i+2}}{\varphi_k}
      \node (cdots1) at (2.20,0.25) {$\cdots$};
      \CaseAState{cak1}{(3.10,0.25)}{a_k}{\varphi_k}
      \CaseARoot{cw1}{(1.55,-1.45)}{a_i}{\neg\varphi_k}
      \node[capsule, fit=(caip11)(caip21)(cdots1)(cak1)] (ccap1) {};
      \node[xlab, left=2pt of ccap1] {$X_{i+1}$};
      \draw[edge] (cw1) -- (caip11);
      \draw[edge] (cw1) -- (caip21);
      \draw[edge] (cw1) -- (cak1);
      \node[font=\large] at (1.55,-2.55) {$M\mid\varphi_k$};
    \end{scope}

    \node[font=\Large] at (10.35,-0.15) {$\overset{\varphi_k}{\Longrightarrow}$};
    \node[font=\large] at (11.45,-0.15) {$\cdots$};
    \node[font=\Large] at (12.55,-0.15) {$\overset{\varphi_k}{\Longrightarrow}$};

    \begin{scope}[shift={(13.80,0)}]
      \CaseAState{cakfinal}{(0,0.25)}{a_k}{\neg\varphi_k}
      \CaseARoot{cwfinal}{(0,-1.45)}{a_i}{\neg\varphi_k}
      \node[capsule, fit=(cakfinal)] (ccapk) {};
      \node[xlab, left=2pt of ccapk] {$X_k$};
      \draw[edge] (cwfinal) -- (cakfinal);
      \node[font=\large] at (0,-2.55) {$M\mid^{k-i}\varphi_k$};
    \end{scope}


    \node[font=\Large] at (-0.35,-5.35) {$\overset{\varphi_k}{\Longrightarrow}$};

    \begin{scope}[shift={(2.15,-5.20)}]
      \node[xlab] at (0,0) {$X_{k+1}(=\emptyset)$};
      \CaseARoot{cwempty}{(0,-1.45)}{a_i}{\neg\varphi_k}
      \node[font=\large] at (0,-2.55) {$M\mid^{k-i+1}\varphi_k$};
    \end{scope}

    \node[font=\Large] at (4.45,-5.35) {$\overset{\varphi_k}{\Longrightarrow}$};
    \node[font=\large] at (5.75,-5.35) {$\cdots$};

  \end{tikzpicture}
  \caption{The case with $\text{type}(w)=a_i$ and $M,w\models E_{X_\ell}$ for some $2\leq i\leq k$ and $i\leq \ell\leq k+1$.}
  \label{fig:0k-valid_but_not_0_kplus1-valid_c}
\end{subfigure}
\caption{$\varphi_k$ is non-trivially $0^k$-valid but not $0^{k+1}$-valid in multi-agent S5 (Lemma~\ref{lem:0k-valid_but_not_0_plus1-validity}).}
\label{fig:Existence_of_0_k-valid_but_not_0_k+1-valid_s5}
\end{figure}
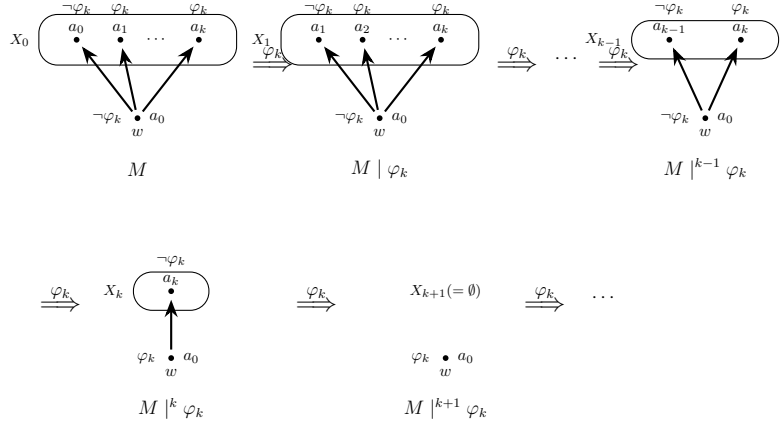
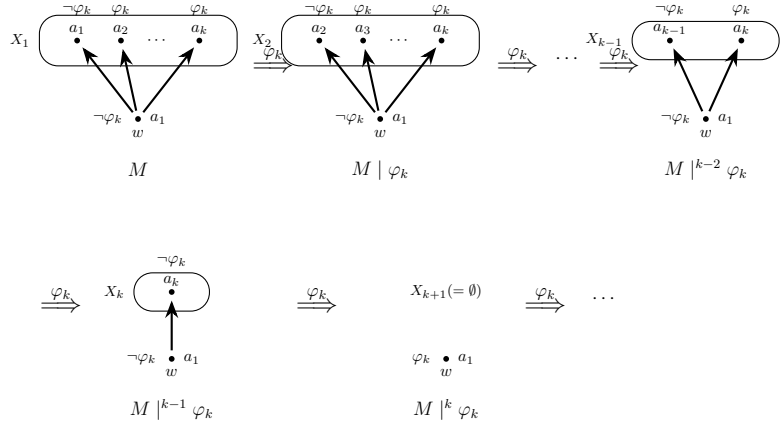
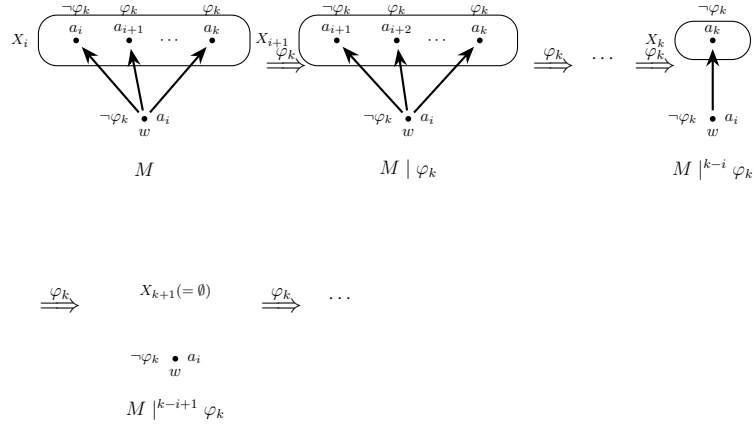

\begin{lemma}[Existence of $101$-validity]\label{lem:101-validity}
  In multi-agent KD45 and S5: There is a non-trivially $101$-valid formula.
\end{lemma}
\begin{proof}
See the footnote in Lemma~\ref{lem:10k-validity_collapse}.
\end{proof}

\section{Discussion}\label{sec:Discussions}
\paragraph*{Interpretation of our results}
In this paper, we proved the following result:
\begin{enumerate}[label=\textup{(\arabic*)}]
    \item In multi-agent K45:
          \begin{equation*}
            S=[0^\omega]\sqcup \bigsqcup_{k\geq 1}[01^k]\sqcup[01^\omega]\sqcup[10^\omega]\sqcup[1^\omega].
          \end{equation*}
    \item In single-agent KD45:
          \begin{equation*}
            S=[0^\omega]\sqcup \bigsqcup_{k\geq 1}[01^k]\sqcup[01^\omega]\sqcup[10]\sqcup[10^\omega]\sqcup[101^\omega ]\sqcup[1^\omega].
          \end{equation*}
    \item In multi-agent KD45 with $|G|\geq 2$ and multi-agent S5:
          \begin{equation*}
            S=[0^\omega]\sqcup\bigsqcup_{k\geq 2}[0^k]\sqcup\bigsqcup_{k\geq 1}[01^k]\sqcup[01^\omega]\sqcup[10]\sqcup[10^\omega]\sqcup[101^\omega]\sqcup[1^\omega].
          \end{equation*}
  \end{enumerate}

These classifications have several implications for the four notions of success, self-refutation, true lies, and impossible lies, as well as on the properties of iterated announcements in general. In this analysis, we restrict our attention mainly to KD45 and S5. 

First, every successful formula remains true forever when initially true. On the other hand, not every impossible lie remains false forever when initially false. So, successful formulas are \emph{stable} while impossible lies are not stable in general.

Second, although not every self-refuting formula becomes false forever when initially true, if it becomes false again, it remains false forever. On the other hand, not every true lie becomes true forever when initially false, and this applies no matter how many times a formula stays true when initially false.

These facts reflect the asymmetry between truthful and false announcements. Roughly speaking, truthful announcements are stable while false announcements are fragile (i.e., truths created by lying are fragile \footnote{Here, ``lying'' is just used as an alternative for ``false announcement''. However, lying that $\varphi$ typically requires the speaker to believe $\lnot\varphi$ and sometimes requires the intention to deceive the listener. Note that neither PAL or BPAL model speakers, per se. For details, see \citep{vanDitmarschEtAl2012,vanDitmarsch2014,SakamaCaminadaHerzig2015,LiVanEijck2022}.}). A more fine-grained perspective would be that truthful announcements are like the win-loss record of a \emph{best-of-three match} for one player when 1 and 0 represent win and loss, respectively: Note that when $\sigma\in S$ starts with 1, the truth stabilizes at 0 once two 0s appear in $\sigma$, and the truth stabilizes at 1 once two 1s appear in $\sigma$. This asymmetry is rooted in the monotonicity of (believed) public announcements, where agents delete false possibilities forever.

The splitting of $[10]$ in K45 into $[10],[10^\omega],[101^\omega]$ in KD45 and S5 occurs because K45 is closed under updates while KD45 and S5 are not closed due to a possible seriality failure. Recall that the formula $\varphi:=(p\land\Diamond p\land\Diamond \lnot p)\lor\Box\bot$ in Lemmas~\ref{lem:10k-validity_collapse} and \ref{lem:101-validity} is non-trivially $101^\omega$-valid. The proof in KD45 went as follows: If $\varphi$ is true, only $p\land\Diamond p\land\Diamond \lnot p$ is true by initial seriality. So, an agent believes $p\land\Diamond p\land\Diamond\lnot p$ after the first announcement. In particular, he believes $p$, which makes $\Diamond\lnot p$ false and hence $p\land\Diamond p\land\Diamond\lnot p$ false. Therefore, $\varphi$ is false after the first announcement. However, the agents's belief becomes inconsistent after the second announcement, making $\Box\bot$ true. This makes $\varphi$ true after the second announcement onwards.  

Another point worth noting is that not all impossible lies are stable in multi-agent KD45 with $|G|\geq 2$ and S5, while all impossible lies are stable in multi-agent K45 and single-agent KD45. We could make three comparisons, (1) multi-agent KD45 vs single-agent KD45, (2) single-agent S5 and single-agent KD45, and (3) multi-agent KD45 vs multi-agent K45. The comparison suggests two distinct sources of instability of impossible lies: factivity already suffices in single-agent S5, whereas in KD45 instability requires the combination of seriality and interaction between multiple agents.   
\paragraph*{PAL and BPAL}
As for logics of public announcements, BPAL is considered to have several advantages over PAL at least in terms of generality and uniformity. As mentioned in the introduction, in BPAL, we can still consider the truth of $\varphi$ at $w$ after an update even if $\varphi$ was initially false while this is impossible in PAL. In particular, the updated model $M_{p}$, for example, must become empty or undefined when $p$ was initially false at all states in $M$. Of course, the definition of $M,w\models[!\varphi]\psi$ as $M,w\models\varphi\Rightarrow M_\varphi,w\models\psi$ allows the evaluation of $[!\varphi]\psi$ at any state in the initial model but this does not solve the problem. 

Although PAL is typically intended for S5, \cite{hollidayIcard2010} required $M,w\models\Diamond\varphi\land\varphi$ as a precondition for public announcements to deal with both KD45 and S5, saying, ``Since we are also working with KD45, we additionally require that $\varphi$ be true at an accessible point, so that $M_\varphi$ is a quasi-partition provided $M$ is.''\footnote{Here, quasi-partition refers to serial, transitive, and Euclidean models.} However, requiring $\Diamond\varphi$ only allows announcements that an agent considers possible and hence cannot deal with announcements that are unexpected or contradict the agents' beliefs.
\paragraph*{Future work}
A natural future direction is to add the common knowledge operator $C_G$ to $\mathcal{L}_{\text{BPAL}}$ since public announcements are inherently connected with common knowledge. We could also consider the classification results for transfinite iterated announcements, in which a formula is announced transfinitely many times through ordinals. 

\backmatter
\bmhead{Acknowledgements}
I thank Ryo Kashima and Koki Okura for their comments and feedback during seminars.

The author used GPT-5.5 Pro and GPT-5.5 for reasoning, coding, drawing figures, and proofreading but the entire manuscript was written by the author, who takes full responsibility for the final content. 
\section*{Statements and Declarations}
\bmhead{Funding}
This research was supported by the Science Tokyo Support Program for Doctoral Students, funded by the Universities for International Research Excellence.
\bmhead{Competing interests}
The author has no competing interests to declare.

\bibliography{bibtex_for_classification_of_sigma_validity}
\end{document}